\RequirePackage[l2tabu,orthodox]{nag}
\documentclass
[11pt,letterpaper]
{article}

\usepackage[utf8]{inputenc}
\usepackage{etex}
\usepackage{xspace,enumerate}
\usepackage[dvipsnames]{xcolor}
\usepackage[T1]{fontenc}
\usepackage[full]{textcomp}
\usepackage[american]{babel}
\usepackage{mathtools}
\usepackage{titling}
\usepackage{amsthm}
\newtheorem{theorem}{Theorem}[section]
\newtheorem*{theorem*}{Theorem}

\newtheorem{proposition}[theorem]{Proposition}
\newtheorem*{proposition*}{Proposition}
\newtheorem{lemma}[theorem]{Lemma}
\newtheorem*{lemma*}{Lemma}
\newtheorem{corollary}[theorem]{Corollary}
\newtheorem*{conjecture*}{Conjecture}
\newtheorem{fact}[theorem]{Fact}
\newtheorem*{fact*}{Fact}

\newtheorem*{hypothesis*}{Hypothesis}

\theoremstyle{definition}
\newtheorem{definition}[theorem]{Definition}
\newtheorem*{definition*}{Definition}

\newtheorem{algorithm}[theorem]{Algorithm}

\newtheorem{model}[theorem]{Model}

\theoremstyle{remark}

\newtheorem*{claim*}{Claim}
\newtheorem{remark}[theorem]{Remark}
\newtheorem*{remark*}{Remark}

\newtheorem*{observation*}{Observation}

\usepackage[
letterpaper,
top=1.2in,
bottom=1.2in,
left=1in,
right=1in]{geometry}
\usepackage{newpxtext} 
\usepackage{textcomp} 
\usepackage[varg,bigdelims]{newpxmath}
\usepackage[scr=rsfso]{mathalfa}
\usepackage{bm} 
\let\mathbb\varmathbb
\usepackage{microtype}
\usepackage[
pagebackref,
colorlinks=true,
urlcolor=blue,
linkcolor=blue,
citecolor=OliveGreen,
]{hyperref}
\usepackage[capitalise,nameinlink]{cleveref}
\crefname{lemma}{Lemma}{Lemmas}
\crefname{fact}{Fact}{Facts}
\crefname{theorem}{Theorem}{Theorems}
\crefname{corollary}{Corollary}{Corollaries}
\crefname{claim}{Claim}{Claims}
\crefname{example}{Example}{Examples}
\crefname{algorithm}{Algorithm}{Algorithms}
\crefname{problem}{Problem}{Problems}
\crefname{definition}{Definition}{Definitions}
\usepackage{paralist}
\usepackage{turnstile}
\usepackage{mdframed}
\usepackage{tikz}
\usepackage{caption}
\DeclareCaptionType{Algorithm}
\usepackage{newfloat}

\newcommand{\Authornote}[2]{}
\newcommand{\FIXME}[1]{}
\newcommand{\Authornotecolored}[3]{}
\newcommand{\Authorcomment}[2]{}
\newcommand{\Authorfnote}[2]{}

\definecolor{forestgreen(traditional)}{rgb}{0.0, 0.27, 0.13}

\usepackage{boxedminipage}
\newcommand{\paren}[1]{(#1)}
\newcommand{\Paren}[1]{\left(#1\right)}

\newcommand{\set}[1]{\{#1\}}
\newcommand{\Set}[1]{\left\{#1\right\}}

\newcommand{\norm}[1]{\lVert#1\rVert}
\newcommand{\Norm}[1]{\left\lVert#1\right\rVert}

\newcommand{\iprod}[1]{\langle#1\rangle}
\newcommand{\Iprod}[1]{\left\langle#1\right\rangle}

\newcommand{\Esymb}{\mathbb{E}}
\newcommand{\Psymb}{\mathbb{P}}

\DeclareMathOperator*{\E}{\Esymb}

\DeclareMathOperator*{\ProbOp}{\Psymb}
\renewcommand{\Pr}{\ProbOp}

\newcommand{\from}{\colon}

\newcommand{\mper}{\,.}
\newcommand{\mcom}{\,,}
\newcommand\bdot\bullet

\DeclareMathOperator{\Tr}{Tr}

\DeclareMathOperator{\poly}{poly}

\newcommand{\N}{\mathbb N}
\newcommand{\R}{\mathbb R}

\newcommand{\cA}{\mathcal A}
\newcommand{\cB}{\mathcal B}

\newcommand{\cD}{\mathcal D}

\newcommand{\cI}{\mathcal I}

\newcommand{\cL}{\mathcal L}

\newcommand{\cN}{\mathcal N}
\newcommand{\cO}{\mathcal O}

\newcommand{\cS}{\mathcal S}

\def\II{\mathbb{I}}

\def\sub{\texttt{Sub}}

\newcommand{\tmu}{\tilde{\mu}}

\renewcommand{\S}{\mathbb{S}}

\renewcommand{\leq}{\leqslant}

\renewcommand{\geq}{\geqslant}
\renewcommand{\ge}{\geqslant}
\let\epsilon=\varepsilon
\numberwithin{equation}{section}
\newcommand\MYcurrentlabel{xxx}
\newcommand{\MYstore}[2]{%
  \global\expandafter \def \csname MYMEMORY #1 \endcsname{#2}%
}
\newcommand{\MYload}[1]{%
  \csname MYMEMORY #1 \endcsname%
}
\newcommand{\MYnewlabel}[1]{%
  \renewcommand\MYcurrentlabel{#1}%
  \MYoldlabel{#1}%
}
\newcommand{\MYdummylabel}[1]{}
\newcommand{\torestate}[1]{%
  \let\MYoldlabel\label%
  \let\label\MYnewlabel%
  #1%
  \MYstore{\MYcurrentlabel}{#1}%
  \let\label\MYoldlabel%
}
\newcommand{\restatetheorem}[1]{%
  \let\MYoldlabel\label
  \let\label\MYdummylabel
  \begin{theorem*}[Restatement of \cref{#1}]
    \MYload{#1}
  \end{theorem*}
  \let\label\MYoldlabel
}
\newcommand{\restatelemma}[1]{%
  \let\MYoldlabel\label
  \let\label\MYdummylabel
  \begin{lemma*}[Restatement of \cref{#1}]
    \MYload{#1}
  \end{lemma*}
  \let\label\MYoldlabel
}
\newcommand{\restateprop}[1]{%
  \let\MYoldlabel\label
  \let\label\MYdummylabel
  \begin{proposition*}[Restatement of \cref{#1}]
    \MYload{#1}
  \end{proposition*}
  \let\label\MYoldlabel
}
\newcommand{\restatefact}[1]{%
  \let\MYoldlabel\label
  \let\label\MYdummylabel
  \begin{fact*}[Restatement of \prettyref{#1}]
    \MYload{#1}
  \end{fact*}
  \let\label\MYoldlabel
}
\newcommand{\restate}[1]{%
  \let\MYoldlabel\label
  \let\label\MYdummylabel
  \MYload{#1}
  \let\label\MYoldlabel
}

\allowdisplaybreaks
\newcommand*{\Id}{\mathrm{Id}}

\newcommand*{\tr}{\mathrm{tr}}
\newcommand*{\zo}{\set{0,1}}

\def\pE{\tilde{\mathbb{E}}}

\def\expec#1#2{{\bf \mathbb{E}}_{#1}[ #2 ]}
\def\expecf#1#2{ \mathop{ \bf \mathbb{E}}_{#1}\left[ #2 \right]}

\def\dim#1{\mathrm{dim} (#1)}

\def\norm#1{\left\| #1 \right\|}

\def\expec#1#2{{\bf \mathbb{E}}_{#1}[ #2 ]}

\title{
  List-Decodable Subspace Recovery: Dimension Independent Error in Polynomial Time
}

\author{
    Ainesh Bakshi\thanks{Carnegie Mellon University} \and Pravesh K. Kothari \thanksmark{1}}

\begin{document}

\pagestyle{empty}


\maketitle
\thispagestyle{empty} 


\begin{abstract}
In list-decodable subspace recovery, the input is a collection of $n$ points $\alpha n$ (for some $\alpha \ll 1/2$) of which are drawn i.i.d. from a distribution $\cD$ with a isotropic rank $r$ covariance $\Pi_*$ (the \emph{inliers}) and the rest are arbitrary, potential adversarial outliers. The goal is to recover a $O(1/\alpha)$ size list of candidate covariances that contains a $\hat{\Pi}$ close to $\Pi_*$. Two recent independent works~\cite{raghavendra2020list,bakshi2020listdecodable} gave the first efficient algorithm for this problem. These results, however, obtain an error that grows with the dimension (linearly in~\cite{raghavendra2020list} and logarithmically in \cite{bakshi2020listdecodable} at the cost of quasi-polynomial running time) and rely on \emph{certifiable anti-concentration} - a relatively strict condition satisfied essentially only by the Gaussian distribution.


In this work, we improve on these results on all three fronts: we obtain \emph{dimension-independent} error in fixed-polynomial running time under less restrictive distributional assumptions. Specifically, we give a $\poly(1/\alpha) d^{O(1)}$ time algorithm that outputs a list containing a $\hat{\Pi}$ satisfying $\Norm{\hat{\Pi} -\Pi_*}_F \leq O(1/\alpha)$. Our result only needs \emph{certifiable hypercontractivity} and bounded variance of degree 2 polynomials -- a condition satisfied by a much broader family of distributions in contrast to certifiable anticoncentration. As a result, in addition to Gaussians, our algorithm applies to uniform distribution on the hypercube and $q$-ary cubes and arbitrary product distributions with subgaussian marginals. Prior work~\cite{raghavendra2020list} had identified such distributions as potential hard examples as such distributions do not exhibit strong enough anti-concentration. When $\cD$ satisfies certifiable anti-concentration, we obtain a stronger error guarantee of $\Norm{\hat{\Pi}-\Pi_*}_F \leq \eta$ for any arbitrary $\eta > 0$ in $d^{O(\poly(1/\alpha) + \log (1/\eta))}$ time. Our proof introduces two new ingredients to SoS based robust estimation algorithms: the use of \textit{certifiable hypercontractivity} to identify low dimensional structure in data and a new exponential error reduction mechanism within SoS. 

\end{abstract}

\clearpage


  \microtypesetup{protrusion=false}
  \tableofcontents{}
  \microtypesetup{protrusion=true}

\clearpage

\pagestyle{plain}
\setcounter{page}{1}

\section{Introduction}
Imagine that you are given points $x_1, x_2, \ldots, x_n \in \R^d$ such that for a fixed constant $\alpha\ll 1/2$, an $\alpha$ fraction are an i.i.d. sample from some isotropic distribution $\cD$ restricted to an unknown subspace $V_*$. The remaining $(1-\alpha)n$ samples - an overwhelming majority - are arbitrary and potentially adversarially chosen \emph{outliers}. Can we learn anything meaningful about the unknown subspace $V_*$ from such a harsh corruption in the input sample?

Uniquely recovering $V_*$ (or any reasonable approximation to it) is clearly impossible - the outliers can contain an $\alpha n$-size i.i.d. sample from the distribution $\cD$ restricted to a different subspace $V'$. A recent line of work on \emph{list-decodable learning} proposed a way forward via a happy marriage between list-decoding from coding theory and algorithmic robust statistics. In this reformulation, the goal is to extract a list of candidates containing an approximate solution, such that the list is of size $O(1/\alpha)$, a fixed constant independent of the dimension $d$. In the context of \textit{subspace recovery}, the list must contain a candidate subspace that approximates the true subspace $V_*$. 

Two recent, concurrent works~\cite{raghavendra2020list,bakshi2020listdecodable}\footnote{One of them is a previous version (\url{https://arxiv.org/abs/2002.05139v1}) of this paper.} discovered the first efficient algorithm for this problem based on the sum-of-squares method. Both results need to assume that the distribution $\cD$ is \emph{certifiably anti-concentrated} - a property known to be satisfied only by Gaussian distributions, and subgaussian-tailed rotationally symmetric distributions so far. The error bounds in both these works\footnote{The work of \cite{raghavendra2020list} gives a $d^{O(1)}$ time algorithm guaranteeing a candidate subspace $V$ with the associated projection matrix satisfying $\Norm{\Pi_V-\Pi_{V_*}}_F\leq O(\sqrt{\dim{V}}/\alpha^{2.5})$. In the same setting, the algorithm from a previous version of this work~\cite{bakshi2020listdecodable} uses $d^{\log d/\alpha^4}$ time to guarantee a candidate subspace $V$ so that $\Norm{\Pi_V-\Pi_{V_*}}_F\leq O(\sqrt{\log \dim(V)}/\alpha)$.} grow with the dimension of the subspace $V$ which can be as large as $\Omega(d)$. 

In a recent work, Karmarkar, Klivans and Kothari~\cite{DBLP:journals/corr/abs-1905-05679} showed that the anti-concentration assumption  made in the above works is necessary for list-decodable linear regression (a special case of list-decodable subspace recovery) by proving an information-theoretic impossibility result for the problem that works even when the distribution of the inlier covariates is uniform on the Boolean hypercube. Thus, one might suspect that anti-concentration of $\cD$ is necessary even for list-decodable subspace recovery. Since no distribution on the hypercube can be anti-concentrated, Raghavendra and Yau~\cite{raghavendra2020list} suggest\footnote{See discussion preceding Lemma 4.1 in~\cite{raghavendra2020list}.}  that such distributions might be hard examples for obtaining dimension-independent error guarantees for list-decodable subspace recovery\footnote{We note that their lower-bound does not apply to the uniform distribution on the hypercube.}.


In this work, we give dramatically improved algorithms for list-decodable subspace recovery that, in particular, show that the above intuition is \emph{not} well-founded. Our algorithm works for a broad class of distributions including the uniform distribution on the discrete hypercube - that is \emph{not} anti-concentrated - while achieving a smaller, \emph{dimension-independent} error in fixed (i.e. exponent independent of $\alpha$) polynomial running time. While this might appear counter-intuitive, we discuss (see paragraph following Fact 1.5 on Page 4) why this result does not violate the information-theoretic lower-bounds of~\cite{DBLP:journals/corr/abs-1905-05679}. When $\cD$ satisfies \emph{certifiable anti-concentration} in addition, we obtain an even better guarantee: a polynomial time algorithm with an \emph{arbitrarily small} error $\eta$ while paying only a $\log(1/\eta)$ factor in the exponent of the run time. At a conceptual level, our main contribution is identifying what appears like the ``right'' distributional assumption that allows obtaining dimension-independent error in list decodable subspace recovery. Our second result is obtained via a new exponential-error-reduction technique within the sum-of-squares framework that is likely to find more applications in the future. 

\paragraph{Context for List-Decodable Learning.} List-decodable learning was first proposed in the context of clustering by Balcan, Blum and Vempala~\cite{DBLP:conf/stoc/BalcanBV08}. In a recent work, Charikar, Steinhardt and Valiant \cite{DBLP:conf/stoc/CharikarSV17} rejuvenated it as a natural model for algorithmic robust statistics. Most recent works in algorithmic robust statistics have focused on the related but less harsh model of where input data is corrupted by an $\epsilon < 1/2$ fraction outliers. This line of work boasts of some remarkable successes including robust algorithms for computing mean, covariance and higher moments of distributions ~\cite{DBLP:conf/focs/LaiRV16,DBLP:conf/focs/DiakonikolasKK016,DBLP:conf/stoc/CharikarSV17,DBLP:journals/corr/abs-1711-11581,DBLP:journals/corr/SteinhardtCV17,DBLP:conf/soda/0002D019,DBLP:conf/icml/DiakonikolasKK017,DBLP:conf/soda/DiakonikolasKK018,DBLP:conf/colt/0002D0W19, DBLP:journals/corr/abs-1711-11581}, clustering mixture models \cite{DBLP:conf/focs/DiakonikolasKS17,KothariSteinhardt17,HopkinsLi17,bakshi2020outlier,diakonikolas2020robustly}  and performing linear regression in the presence of a small $\epsilon$ fraction of adversarial outliers ~\cite{DBLP:conf/focs/DiakonikolasKS17,DBLP:conf/colt/KlivansKM18,DBLP:conf/icml/DiakonikolasKK019,DBLP:journals/corr/abs-1802-06485,DBLP:journals/corr/abs-1905-05679,bakshi2020robust,zhu2020robust,cherapanamjeri2020optimal}. 

While the success hasn't been of the same scale, there has been quite a bit of progress on list-decodable learning that surmount the challenges that arise in dealing with overwhelmingly corrupted data. Recent sequence of works have arrived at a blueprint using the \emph{sum-of-squares method} for list-decodable estimation with applications to list-decodable mean estimation~\cite{diakonikolas2018list,KothariSteinhardt17} and linear regression~\cite{DBLP:journals/corr/abs-1905-05679,raghavendra2020listreg}.

\paragraph{Why List-Decodable learning?} List-decodable learning is a strict generalization of related and well-studied \emph{clustering} problems (for e.g., list-decodable mean estimation generalizes clustering spherical mixture models, list-decodable regression generalizes mixed linear regression). In our case, list-decodable subspace recovery generalizes the well-studied problem of subspace clustering where given a mixture of $k$ distributions with covariances non-zero in different subspaces, the goal is to recover the underlying $k$ subspaces ~\cite{agrawal1998automatic,cheng1999entropy,goil1999mafia,procopiuc2002monte,aggarwal2000finding}. Algorithms in this model thus naturally yield robust algorithms for the related clustering formulations. In contrast to known results, such algorithms allow ``partial recovery'' (e.g. for example recovering $k-1$ or fewer clusters) even in the presence of outliers that garble up one or more clusters completely.

Another important implication of list-decodable estimation is algorithms for the \emph{small outlier} model that work whenever the fraction of inliers $\alpha > 1/2$ - the information-theoretic minimum for unique recovery. As a specific corollary, we obtain an algorithm for (uniquely) estimating the subspace spanned by the inlier distribution $D$ whenever $\alpha > 1/2$. We note that if $\alpha$ is sufficiently close to $1$, such a result follows from outlier-robust covariance estimation algorithms~\cite{DBLP:journals/corr/DiakonikolasKKL16,DBLP:conf/focs/LaiRV16,DBLP:conf/colt/0002D0W19}. While prior works do not specify precise constants, all known works appear to require $\alpha$ at least $> 0.75$.

\subsection{Our Results}
We are ready to formally state our results. Our results apply to input samples generated according to the following model:

\begin{model}[Robust Subspace Recovery with Large Outliers]
For $0< \alpha<1$ and $r < d$, 
let $\Pi_* \in \mathbb{R}^{d \times d}$ be a projector to a subspace of dimension $r \leq d$ and let $\cD$ be a distribution on $\R^d$ with mean $\mu_*$ and covariance $\Pi_*$. Let $\sub_{\cD}(\alpha, \Pi_*)$ denote the following probabilistic process to generate $n$ samples, $x_1, x_2 \ldots x_n$ with $\alpha n$ inliers $\cI$ and $(1-\alpha)n$ outliers $\cO$:
\begin{enumerate}
    \item Construct $\cI$ by choosing $\alpha n$ i.i.d. samples from $\cD$.
    \item Construct $\cO$ by choosing the remaining $(1-\alpha)n$ points arbitrarily and potentially adversarially w.r.t. the inliers.
\end{enumerate}
\end{model}

\begin{remark} We will mainly focus on the case when $\mu_* = 0$. The case of non-zero $\mu_*$ can be easily reduced to the case of $\mu_* = 0$ by modifying samples by randomly pairing them up and subtracting off samples in each pair (this changes the fraction of inliers from $\alpha$ to $\alpha^2$).
\end{remark}

Our main result is a \emph{fixed} (i.e. exponent of the polynomial does not depend on $\alpha$) polynomial time algorithm with \emph{dimension-independent} error in \emph{Frobenius norm} - the strongest notion of closeness that implies other guarantees such as the principal angle and spectral distance between subspaces - for list-decodable subspace recovery that succeeds whenever $\cD$ has certifiably hypercontractive degree-$2$ polynomials:

\begin{definition}[Certifiably Hypercontractivity]
A distribution $\cD$ on $\R^d$ is said to have $(C,2h)$-certifiably hypercontractive polynomials if there is a degree-$2h$ sum-of-squares proof (see Subsection \ref{subsec:sos_proofs}) in the $d \times d$ matrix-valued indeterminate $Q$ of the following inequality:
\[
\E_{x\sim \cD} \Paren{x^{\top} Q x -\E_{x\sim \cD}x^{\top} Q x}^{2h} \leq  (Ch)^{2h} \Paren{\E_{x\sim \cD} \Paren{x^{\top} Q x -\E_{x\sim \cD}x^{\top} Q x}^{2}}^h\mper
\]
\end{definition}

Many natural distributions are certifiably hypercontractive including linear transforms of uniform distribution on the Boolean hypercube and unit sphere, Gaussian distributions, and product distributions with subgaussian marginals. In particular, the set of certifiably hypercontractive distributions is strictly larger than the currently known list of certifiably anti-concentrated distributions (that essentially only holds for rotationally symmetric distributions with sufficiently light tails).

We are now ready to state our main result. 

\begin{theorem} \label{thm:main}
Let $\Pi_*$ be a projection matrix for a subspace of dimension $r$. 
Let $\cD$ be a distribution with mean $0$, covariance $\Pi_*$, and certifiably $(C, 8)$-hypercontractive  polynomials. 

Then, there exists an algorithm takes as input $n = n_0 \geq \Omega\left((d\log(d)/\alpha)^{16}\right)$ samples from $\sub_{\cD}(\alpha, \Pi_*)$ and in $O(n^{18})$ time, outputs a list $\cL$ of $O(1/\alpha)$  projection matrices such that with probability at least $0.99$ over the draw of the sample and the randomness of the algorithm, there is a $\hat{\Pi} \in \cL$ satisfying $\|\hat{\Pi} - \Pi_*\|_F \leq O(1/\alpha)$. 
\end{theorem}

As an immediate corollary, this gives an algorithm for list-decodable subspace recovery when $\cD$ is Gaussian, uniform on the unit sphere, uniform on the discrete hypercube/$q$-ary cube, product distribution with subgaussian marginals and their affine transforms.
\paragraph{Discussion and Comparison with Prior Works} 
Theorem \ref{thm:main}  improves on the previous version of this paper~\cite{bakshi2020listdecodable} and ~\cite{raghavendra2020list} in running time, error guarantees and the generality of the distribution $\cD$. In particular, it strictly improves on the work of Raghavendra and Yau who gave an error guarantee of $O(r/\alpha^5)$ in polynomial time by relying on certifiable anti-concentration.\footnote{The results in~\cite{raghavendra2020list} handle a small amount of additive noise. The algorithm in this paper can be extended to handle a similar amount of noise but we do not focus on that aspect in this paper.} It also improves on the guarantee in a previous version of this work for Gaussians that relied on certifiable anti-concentration and an exponential error reduction method to give an error of $O(\log (r)/\alpha)$ in $d^{O(\log r/\alpha^4)}$ time. Unlike the present work, both these algorithms provably cannot extend to the uniform distribution on the hypercube.  

A priori, our result might appear surprising and  almost too-good-to-be-true. Indeed, prior works identified anti-concentration as a information-theoretic necessary condition on $\cD$ for list-decodable regression (a special case of list-decodable subspace recovery) to be feasible. Specifically, Karmalkar, Klivans and Kothari~\cite{DBLP:journals/corr/abs-1905-05679} show:

\begin{fact}[Theorem 6.1, Page 19 in~\cite{DBLP:journals/corr/abs-1905-05679}] \label{fact:lower-bound-kkk}
For any constant $\alpha>0$, there exists a distribution $\cD$ (uniform distribution on $\zo^n$) that is $(\alpha+\epsilon)$-anti-concentrated for every $\epsilon>0$ but there is no algorithm for $\alpha/2$-approximate list-decodable subspace recovery with rank $r=d-1$ that outputs a list of size $<d$.
\end{fact}

On the other hand, note that discrete product distributions such as uniform distribution on the hypercube/$q$-ary cube satisfy certifiable hypercontractivity (see ~\cite{DBLP:conf/soda/KauersOTZ14}) so our Theorem~\ref{thm:main} applies. This is not a contradiction because of the error guarantees - observe that the Frobenius error bound of $O(1/\alpha)$ provided by Theorem~\ref{thm:main} translates to a $\ell_2$-norm bound of $O(1/\alpha)$ for linear regression. This is not meaningful for unit vectors, whenever $\alpha \leq 1/2$, since even a random unit vector achieves an error of at most $\sqrt{2}$ in this setting. On the other hand, for subspace recovery, this is a non-trivial guarantee whenever the dimension and the co-dimension of the unknown subspace are $\gg 1/\alpha$. 

\paragraph{High-Accuracy Subspace Recovery.} 
Our first result naturally raises the question of algorithms obtaining arbitrarily tiny error (instead of $O(1/\alpha)$). For sufficiently small errors ($\ll 1$), $\cD$ must necessarily be anti-concentrated, given the lower-bound from Fact~\ref{fact:lower-bound-kkk} above. Our next result confirms that \textit{certifiable anti-concentration} of $\cD$ is sufficient to obtain an arbitrarily small error while still maintaining a list-size of an absolute constant (but of size $1/\alpha^{O(\log (1/\alpha))}$) independent of the dimension. 

To state our result, we first recall certifiable anti-concentration. 

\begin{definition}[Certifiable Anti-Concentration]
A zero-mean distribution $D$ with covariance $\Sigma$ is $2t$-certifiably $(\delta,C\delta)$-anti-concentrated if there exists a univariate polynomial $p$ of degree $t$ such that there is a degree $2t$ sum-of-squares proof in variable $v$ of the following inequalities:
\begin{enumerate}
\item $\norm{v}^{2t-2}_2\iprod{x,v}^2 + \delta^2 p^2\left(\iprod{x,v}\right) \geq \frac{\delta^2 \norm{\Sigma^{1/2} v}^{2t}_2}{2}$.
\item $\E_{ x \sim D} p^2\paren{\iprod{x,v}} \leq C\delta \norm{\Sigma^{1/2} v}_2^{2t}$.
\end{enumerate}
A subset $\cS \subseteq  \R^d$ is $2t$-certifiably $(\delta,C\delta)$-anti-concentrated if the uniform distribution on $\cS$ is $2t$-certifiably $(\delta,C\delta)$-anti-concentrated.
\label{def:certifiable-anti-concentration-homogenous}
\end{definition}

Gaussian distributions and spherically symmetric distributions with subgaussian tails are $O(1/\delta^2)$-certifiably $(2,\delta)$-anti-concentrated for every $\delta > 0$ (see Section~\ref{sec:certifiable_anti_conc}). 

\begin{theorem}\label{thm:main-2}
Let $\Pi_*$ be a projector to a subspace of dimension $r$. 
Let $\cD$ be a $k$-certifiably $(C,\alpha/2C)$-anti-concentrated distribution with certifiably $C$-hypercontractive degree 2 polynomials. 

Then, there exists an algorithm that takes as input $n = n_0 \geq (d\log(d)/\alpha)^{O(k)}$ samples from $\sub_{\cD}(\alpha, \Pi_*)$ and in $n^{O(k + \log(1/\eta))}$ time, outputs a list $\cL$ of $O(1/\alpha^{\log k + \log (1/\eta)})$ projection matrices such that with probability at least $0.99$ over the draw of the sample and the randomness of the algorithm, there is a $\hat{\Pi} \in \cL$ satisfying $\|\hat{\Pi} - \Pi_*\|_F \leq \eta$. 
\end{theorem}
The proof of Theorem~\ref{thm:main-2} is based on new argument using certifiable anti-concentration that bootstraps our first result with an exponential error reduction mechanism within the sum-of-squares proof system. This improves on the result in a previous version of this work that gave a $d^{O(\log d/\alpha^4)}$ algorithm with a dimension-dependent list size of $d^{O(\log (1/\alpha)}$ based on a somewhat complicated pruning procedure.  

Using $O(1/\delta^2)$-certifiable $(\delta,C\delta)$-anti-concentration of Gaussians and spherically symmetric distribution with subgaussian tails, we obtain:
\begin{corollary}\label{cor:gaussians}
Let $\Pi_*$ be a projector a subspace of dimension $r$. 
Let $\cD$ be a mean $0$ Gaussian or a spherically symmetric distribution with subgaussian tails with covariance $\Pi_*$. 

Then, there exists an algorithm that takes as input $n = n_0 \geq (d\log(d)/\alpha^2)^{O(1/\alpha^2)}$ samples from $\sub_{\cD}(\alpha, \Pi_*)$ and in $n^{ \log(1/\alpha \eta)/\alpha^4}$ time, outputs a list $\cL$ of $O(1/\alpha^{\log 1/\alpha + \log (1/\eta)})$ projection matrices such that with probability at least $0.99$ over the draw of the samples and the randomness of the algorithm, there is a $\hat{\Pi} \in \cL$ satisfying $\|\hat{\Pi} - \Pi_*\|_F \leq \eta$. 
\end{corollary}

\paragraph{Further uses of exponential error reduction.} Our exponential error reduction method is likely to be of wider use. As an example, we observe the following immediate consequence to list-decodable linear regression by  obtaining an improved running time (with a large constant list-size) as a function of the target accuracy. 

\begin{corollary}[List-Decodable Regression] 
Let $\cD$ be $k$-certifiably $(\alpha/2C)$-anti concentrated distribution with mean $0$ and covariance $I$.
Then, there exists an algorithm that takes as input $n = n_0 \geq (d\log(d)/\alpha)^{\widetilde{O}(k)}$ labeled samples where an $\alpha n$ samples $(x,y)$ are i.i.d. with $x \sim \cD$ and $y = \iprod{\ell_*,x}$ for some unknown, unit vector $\ell_*$ and outputs a list $\cL$ of $O(1/\alpha^{O(\log (k) + \log (1/\eta))})$ regressors such that with probability at least $0.99$ over the draw of the samples and the randomness of the algorithm, there is a regressor $\hat{\ell} \in \cL$ satisfying $\|\hat{\ell} - \ell_*\|^2_2 \leq \eta$. The algorithm has time complexity at most $n^{O(k+  \log(1/ \eta ))}$.
\end{corollary}

Prior works~\cite{DBLP:journals/corr/abs-1905-05679,raghavendra2020listreg} needed $n^{O(k^2/\eta^2)}$ time but computed a smaller list of size $O(1/\alpha)$).

\subsection{Previous Versions and Concurrent Work}
A previous version of this work~\cite{bakshi2020listdecodable} appeared concurrently with ~\cite{raghavendra2020list} and gave a $d^{O(\log r)/\alpha^4}$ time algorithm to output a $O(1/\alpha)$ size list that contains a candidate projection matrix that is $O(\frac{\log r}{\alpha})$-Frobenius close to the rank $r$ projection matrix of the true subspace. The algorithm worked whenever the inlier distribution $\cD$ satisfies \emph{certifiable anti-concentration}. This version of the work combines the ideas in~\cite{bakshi2020listdecodable} with multiple new insights to obtain the improved results. 

\subsection{Related Work}
Our work builds on a sequence of works on robust statistics~\cite{KothariSteinhardt17,DBLP:journals/corr/abs-1711-11581,HopkinsLi17,DBLP:conf/colt/KlivansKM18,DBLP:journals/corr/abs-1905-05679,raghavendra2020listreg} using the sum-of-squares method (see \cite{TCS-086} for an exposition). This is an offshoot of a line of work on using sum-of-squares method for algorithm design~\cite{DBLP:journals/corr/BarakM15,MR3388192-Barak15,DBLP:journals/corr/MaSS16,DBLP:conf/colt/HopkinsSS15,DBLP:conf/colt/PotechinS17,DBLP:conf/stoc/BarakKS17}. 

The improvements over the prior version of this paper are directly related to the recent work on outlier-robust clustering of non-spherical mixtures~\cite{bakshi2020outlierrobust}  that abstracted out certifiable hypercontractivity of degree 2 polynomials (along with certifiable anti-concentration) as sufficient conditions (\cite{DBLP:journals/corr/abs-2005-06417} independently obtained related results) for clustering non-spherical mixture models. In particular, our proof of the first result is inspired by the proof of Lemma~4.20 in~\cite{bakshi2020outlierrobust}. Our main insight here is to obtain similar conclusions to those in Lemma~4.20 in~\cite{bakshi2020outlierrobust} without certifiable anti-concentration by exploiting that the inlier covariances are projection matrices. 

\paragraph{Subspace Clustering.} A closely related (and formally easier\footnote{One can think of input to subspace clustering as the special case in list-decodable subspace recovery where the input sample is a mixture of $k = 1/\alpha$ distributions each with a covariance restricted to some subspace.}) problem to list-decodable subspace recovery is subspace clustering~\cite{elhamifar2013sparse,parsons2004subspace,soltanolkotabi2014robust}. Known algorithms with provable guarantees for this problem either require runing time exponential in the ambient dimension, such as RANSAC~\cite{fischler1981random}, algebraic subspace clustering~\cite{vidal2005generalized} and spectral curvature clustering~\cite{liu2012robust},  or require the co-dimension to be a constant fraction of the ambient dimension~\cite{candes2013phaselift, lu2012robust, tsakiris2017hyperplane, tsakiris2018dual}.

\paragraph{Robust Subspace Recovery.}
Our setting also superficially resembles \emph{robust subspace recovery} (see \cite{lerman2018overview} for a survey), where the goal is to recover a set of inliers that span a single low-dimensional space. In this setting, $\alpha$ is assumed to be close to $1$. Prior works on this problem identify some tractable special cases (see \cite{vaswani2018static}) while no provable guarantees are known for the general setting. 
Further, Hardt and Moitra \cite{DBLP:conf/colt/HardtM13} (see also the recent work of Bhaskara, Chen, Perreault and Vijayraghavan~\cite{DBLP:journals/corr/abs-1811-12361})  provide a polynomial time randomized algorithm, where both the inliers and outliers are required to be in general position and their algorithm works as long as the inliers constitute an $\alpha = r/d$ fraction, where $r$ is the rank of the subspace and $d$ is the ambient dimension. This is contrast to our work where the outliers are completely arbitrary and potentially adversarial with respect to the inliers.

\newcommand{\dtv}{d_{\mathrm{TV}}}
\section{Technical Overview}
\label{sec:overview}
In this section, we give a bird's eye view of our algorithm and analysis. 
Our algorithm and rounding follow the framework built for list-decodable learning in ~\cite{DBLP:journals/corr/abs-1905-05679,raghavendra2020list}. Let us formalize the setting before highlighting the key new ideas in our work. 

In the list-decodable subspace recovery problem, our input is a collection of samples $x_1, x_2, \ldots, x_n \in \R^d$, an $\alpha n$  of which are drawn i.i.d. from a distribution $\cD$ with mean $0$ and unknown projective covariance $\Pi_*$ of rank $k$. The main idea of the algorithm is to encode finding the "inliers" in the input sample via a polynomial program. To do this, we introduce variables $w_1, w_2,\ldots, w_n$ that are supposed to indicate the samples that correspond to the inliers. Thus, we force $w_i^2 = w_i$ (i.e. $w_i \in \{0,1\}$) and $\sum_{i\leq n} w_i = \alpha n$ as constraints on $w$. We also introduce a variable $\Pi$ that stands for the covariance of the inliers and add constraints that force it to be a projection matrix.  To this end, it suffices to constraint $\Pi^2 = \Pi$ and $\Pi \succeq 0$. Further, we require that each of the samples indicated by $w$ are in the subspace described by $\Pi$: $w_i (\Pi x_i -x_i) = 0$ for every $i$. 

Recall, an adversary can create multiple rank-$k$ subspaces that satisfy all the aforementioned constraints, and a priori, a solution to the above polynomial program need not tell us anything about the \textit{true} inliers. Therefore, we must force $w$ to share some property that $\cD$ satisfies so that we can guarantee a solution to the program contains some information about the inliers. What property should this be? In the context of list-decodable regression ~\cite{DBLP:journals/corr/abs-1905-05679,raghavendra2020listreg}, it turns out that it was both necessary and sufficient (up to the additional qualifier of "certifiability") for $w$ (and $\cD$) to be anti-concentrated. Anti-concentration is also \emph{sufficient} to get some guarantees for subspace recovery as shown in~\cite{bakshi2020listdecodable,raghavendra2020list}. Is it necessary? And if not, is there a property satisfied by a larger class of distributions that might be sufficient? 

\subsection{Dimension Independent Error via Certifiable Hypercontractivity} 
Our key idea is to observe that anti-concentration is in fact \emph{not} necessary for the list-decodable subspace recovery. Instead, all we need is (SoS certifiable) hypercontractivity: a property that controls the growth of moments of degree-2 polynomials. Specifically, $\cD$ has $C$-hypercontractive degree-$2$ polynomials if for every $d \times d$ matrix $Q$, 
\[
\E_{x \sim \cD} \Paren{x^{\top} Qx -\tr(Q)}^{2h} \leq (Ch)^h \Paren{\E_{x \sim \cD} \Paren{x^{\top} Qx}^{2}}^h \mper
\] 

This is a polynomial inequality in the $d \times d$ matrix-valued indeterminate $Q$. \emph{Certifiable} hypercontractivity means that in addition, the above inequality has a degree-$2h$ SoS proof in the variable $Q$. We discuss SoS proofs in more detail in Subsection~\ref{subsec:sos_proofs}.  

Certifiable hypercontractivity is satisfied by a large class of distributions such as Gaussians and product distributions with subgaussian marginals including those that are not anti-concentrated, such as uniform distribution on the Boolean hypercube and $q$-ary cubes. We note that this is a generalization of \emph{certifiable subgaussianity} for linear polynomials used in several prior works~\cite{DBLP:journals/corr/abs-1711-11581,KothariSteinhardt17,HopkinsLi17,DBLP:conf/colt/KlivansKM18}. 

This leads us to the following system of polynomial constraints that we use in our first result:
\begin{equation}
\label{eqn:constraint_intro}
  \cA_{w,\Pi}\colon
  \left \{
    \begin{aligned}
      &&
      \textstyle\sum_{i\in[n]} w_i
      &= \alpha n\\
      &\forall i\in [n].
      & w_i \Pi x_i
      & = w_i x_i \\
      &\forall i\in [n].
      & w_i^2
      & = w_i \\
      &
      & \Pi^2
      & = \Pi\\
      &\forall Q 
      &\frac{1}{\alpha n} \sum_{i \leq n} w_i \Paren{x_i^{\top} Q x_i}^{2t} 
      &\leq (Ct)^t \Paren{\frac{1}{\alpha n} \sum_{i \leq n} w_i \Paren{x_i^{\top} Q x_i}^2}^t\\
      &\forall Q 
      &\frac{1}{\alpha n} \sum_{i \leq n} w_i \Paren{x_i^{\top} Q x_i-\tr(Q\Pi)}^{2} 
      &\leq C \Norm{\Pi Q\Pi}_F^2\mper
    \end{aligned}
  \right \}
\end{equation}
The last two constraints encode certifiable hypercontractivity and bounded variance of degree $2$ polynomials on the samples indicated by $w$ and might give one a pause - they form a system of infinitely many constraints, one for each $Q$. However, we employ the (by now) standard idea of succinctly representing such infinite system of constraints efficiently (i.e. at most a polynomial number of constraint with exponent linear in the degree of the SoS proof) by exploiting that the constraints admit SoS proofs for the inlier distribution. This technique was first discovered in~\cite{DBLP:journals/corr/abs-1711-11581} and \cite{HopkinsLi17} and has been subsequently used in several works - see Chapter 4.3 in \cite{TCS-086} for a detailed exposition. 

\paragraph{Certifiable Hypercontractivity and bounded variance suffices.} 
Our key contribution is a proof that a solution to the SoS relaxation of $\cA_{w,\Pi}$ \eqref{eqn:constraint_intro} , with a "high-entropy" constraint that we discuss later, is enough to list-decode the unknown $\Pi_*$. The key to this lemma is a SoS version of the following statement that relates total variation distance to parameter distance for certifiably hypercontractive distribution and may be of independent appeal:

\begin{fact}[Informal, follows from Lemma~\ref{lem:close-projectors}] \label{fact:main-tv-to-parameter}
Let $\cD,\cD'$ be distributions on $\R^d$ such that $\dtv(\cD,\cD') \leq 1-\alpha$. Suppose $\cD,\cD'$ have projective covariances $\Sigma,\Sigma'$ and have hypercontractive degree 2 polynomials.  Then, $\Norm{\Sigma-\Sigma'}_F \leq O(1/\alpha)$. 
\end{fact}

This fact says that two certifiably hypercontractive distributions that have TV distance of at most $1-\alpha$, i.e., an overlap of at least $\alpha$, must have covariances that are close in Frobenius norm up to a dimension-independent bound of $O(1/\alpha)$. The proof of this lemma relies on the use of SoS Hölder's inequality (see Fact \ref{fact:sos-holder}). At a high-level, a similar fact is employed in Lemma 4.20~\cite{bakshi2020outlierrobust} in the context of outlier-robust clustering and our work is directly inspired by it. 

To see how Fact \ref{fact:main-tv-to-parameter} is useful, we observe that the uniform distribution on a large enough sample from a certifiably hypercontractive distribution is itself certifiably hypercontractive (see Lemma \ref{lem:cert_hyper_sampling}). Thus, the uniform distribution over the set of inliers, $\cI$, satisfies certifiable hypercontractivity. Via appropriately chosen polynomial constraints, we can force $w$ to indicate a certifiably hypercontractive distribution. Now, observe that if the subsample indicated by $w$ intersects $\cI$ in $\alpha |\cI|$ points, then the TV distance between uniform distribution on $\cI$ and the subsample indicated by $w$ is at most $1-\alpha$ so we can apply the Fact \ref{fact:main-tv-to-parameter} to conclude that $\Norm{\Pi-\Pi_*}_F \leq O(1/\alpha)$. 

The above plan works as long as the set indicated by $w$ has a substantial intersection with $\cI$. But, we cannot quite force $w$ to have a non-zero intersection with $\cI$ at all. Indeed, the outliers may themselves contain another $\alpha n$ points that are i.i.d. sample from $\cD$ with a different projector as a covariance and $w$ could be an indicator of that set instead. However, this issue can be rectified following the prior works on list-decodable learning via SoS, that use high-entropy constraints on the "pseudo-distributions" to simulate (see the overview of ~\cite{DBLP:journals/corr/abs-1905-05679} for an informal discussion) the property of the subsample indicated by $w$ having a large intersection with $\cI$.

The above discussion gives us our final outline of the algorithm: compute a  max-entropy  ``pseudo-distribution'' of large enough constant degree on $w$ satisfying $\cA_{w,\Pi}$ and run the ``rounding-by-votes'' procedure (as in ~\cite{DBLP:journals/corr/abs-1905-05679}) to output a list. 
\subsection{High-Accuracy Subspace Recovery via Certifiable Anti-concentration}
Our next result (Theorem \ref{thm:main-2}) allows obtaining a constant-size list that contains a projection matrix $\Pi$ that is $\eta$-close to $\Pi_*$ in Frobenius distance, for any $\eta>0$,  in time $d^{O(\log^2(1/\eta) +  1/\alpha^8)}$. For this result, we additionally require that $\cD$ is \textit{certifiably anti-concentrated}. 

Certifiable anti-concentration was recently formulated and used for designing algorithms for list-decodable regression~\cite{DBLP:journals/corr/abs-1905-05679,raghavendra2020listreg} and for obtaining the first efficient algorithm for outlier-robust clustering of non-spherical mixtures~\cite{bakshi2020outlierrobust,DBLP:journals/corr/abs-2005-06417}. Recall that anti-concentration captures the idea that any linear projection of the distribution $\cD$ cannot lie in an interval of length at most $\eta$ with probability more than $\eta/\sigma$ where $\sigma$ is the variance of $\cD$ in the projected direction. In the SoS proof system, this property is formulated via a low-degree polynomial $p$ that approximates the indicator function on a $\delta$-length interval around $0$ (see Definition \ref{def:certifiable-anti-concentration-homogenous} for a formal statement). 

\paragraph{Information-Theoretic Exact Recovery.}
We begin by showing that given two distributions that have non-negligible overlap in TV distance and are anti-concentrated, the subspaces spanned by their respective covariances must be \textit{identical}.  

\begin{fact}[Informal, Follows from Proposition~\ref{prop:high-intersection-same-subspace}] \label{fact:anti-conc-subspace}
Let $\cD,\cD'$ be two distributions are anti-concentrated with projective covariances $\Pi,\Pi'$ respectively.
Suppose $\dtv(\cD,\cD') \leq 1-\alpha$ for some $\alpha > 0$.
Then, $\Pi = \Pi'$.
\end{fact}

To use Fact \ref{fact:anti-conc-subspace}, we show that sampling preserves anti-concentration (see Lemma \ref{lem:certifiable_anti_conc_sampling}). Thus, in particular, if the uniform distribution over the inliers, $\cI$, and the uniform distribution over the samples indicated by $w$ intersect, we can conclude $\Pi = \Pi_*$. In order to obtain an efficient algorithm, we need to "SoSize" this proof. 
 
\paragraph{Real life vs SoS: Bootstrapping to remove dimension-dependence.} The SoS version of Fact \ref{fact:anti-conc-subspace} faces a difficulty that has been encountered in every usage~\cite{DBLP:journals/corr/abs-1905-05679,raghavendra2020listreg,bakshi2020outlierrobust,DBLP:journals/corr/abs-2005-06417,raghavendra2020list} of \emph{certifiable} anti-concentration\footnote{See paragraph title "Combining Upper and lower bounds: real life vs sos" on page 7 and paragraph titled "Digression" on page 27 in ~\cite{bakshi2020outlierrobust} for a discussion of this technical issue.}. 

Informally, in order to SoSize the proof of Fact~\ref{fact:anti-conc-subspace} above, we need a "rough" \emph{a priori} upper-bound on $\Norm{\Pi-\Pi_{*}}_F$. 
In a previous version of this paper, we used the trivial upper bound of $\sqrt{2d}$ on the above Frobenius norm that follows simply by using that both $\Pi$ and $\Pi_*$ are projectors. This was the bottleneck that lead to both the super-constant list-size and super-polynomial running time of that algorithm. 
Our key observation in this work is that we can use the guarantees from Theorem \ref{thm:main} above in order to obtain a dimension-independent constant as the  rough upper bound. 

\paragraph{Exponential Error Reduction.}
We observe that certifiable anti-concentration along with Theorem \ref{thm:main} readily yield an algorithm with list-size $1/\alpha^{\poly(1/\eta)}$ and running time $d^{\poly(1/\eta)\poly(1/\alpha)}$. However, we can improve the exponent in both list size and running time exponentially using a new \textit{powering trick} that we introduce in this work. This idea allows us to replace the polynomial factors in $1/\eta$ appearing in the exponent with $\log(1/\eta)$ instead. As an immediate corollary, we obtain algorithms with exponentially better dependence in the accuracy parameter for list-decodable regression as well. We believe that our powering technique would be widely applicable in robust statistics and more generally, adds to a growing toolbox \textit{sum-of-squares} proofs.

At a high level, we obtain a degree-$t$ \textit{sum-of-squares} proof of the following implication: given $a, b\geq 0$, such that $(a-b)Z \leq 0$, $(a^t - b^t)Z \leq 0$ in the indeterminates $a,b$. We apply the powering trick to the inequality $(w(\cI) - C\delta)\Norm{\Pi - \Pi_*}^k_F \leq 0$, where $w(\cI)$ is the average weight $w$ places on the inliers and $C\delta$ is the slack term from certifiable anti-concentration. 


We stress that \emph{information-theoretically}, anti-concentration, all by itself, (via Fact~\ref{fact:anti-conc-subspace}) implies perfect, zero-error subspace recovery with list-size $O(1/\alpha)$ (see a proof sketch in Appendix \ref{sec:exact_recovery}). However, it is
the "SoS-ization" step (to obtain an efficient algorithm) that leads to the weaker, dimension-dependent guarantees in the previous works. Indeed, it is somewhat intriguing from the perspective f sum-of-squares based algorithm design that we gain in both running time and list-size bounds by introducing a new analytic condition, \textit{certifiable hypercontractivity} of degree-$2$ polynomials, that is not a priori needed in our certifiability proofs.

\section{Preliminaries}
\label{sec:preliminaries}


Throughout this paper, for a vector $v$, we use $\norm{v}_2$ to denote the Euclidean norm of $v$. For a $n \times m$ matrix $M$, we use $\norm{M}_2 = \max_{\norm{x}_2=1} \norm{Mx}_2$ to denote the spectral norm of $M$ and $\norm{M}_F = \sqrt{\sum_{i,j} M_{i,j}^2}$ to denote the Frobenius norm of $M$. For symmetric matrices we use $\succeq$ to denote the PSD/Löwner ordering over eigenvalues of $M$. For a $n \times n$, rank-$r$ symmetric matrix $M$, we use $U\Lambda U^{\top}$ to denote the Eigenvalue Decomposition, where $U$ is a $n \times r$ matrix with orthonormal columns and $\Lambda$ is a $r \times r$ diagonal matrix denoting the eigenvalues. We use $M^{\dagger} = U \Lambda^{\dagger} U^{\top} $ to denote the Moore-Penrose Pseudoinverse, where $\Lambda^{\dagger}$ inverts the non-zero eigenvalues of $M$. If $M \succeq 0$, we use $M^{\dagger/2} = U \Lambda^{\dagger/2} U^{\top}$ to denote taking the square-root of the inverted non-zero eigenvalues. We use $\Pi = U U^{\top}$ to denote the Projection matrix corresponding to the column/row span of $M$. Since $\Pi = \Pi^2$, the pseudo-inverse of $\Pi$ is itself, i.e. $\Pi^{\dagger} = \Pi$.

In the following,  we define pseudo-distributions and sum-of-squares proofs. Detailed exposition of the sum-of-squares method and its usage in average-case algorithm design can be found in ~\cite{TCS-086} and the lecture notes~\cite{BarakS16}.

\subsection{Pseudo-distributions}
\label{subsec:psuedo_dist_prelims}

Let $x = (x_1, x_2, \ldots, x_n)$ be a tuple of $n$ indeterminates and let $\R[x]$ be the set of polynomials with real coefficients and indeterminates $x_1,\ldots,x_n$.
We say that a polynomial $p\in \R[x]$ is a \emph{sum-of-squares (sos)} if there are polynomials $q_1,\ldots,q_r$ such that $p=q_1^2 + \cdots + q_r^2$.

Pseudo-distributions are generalizations of probability distributions.
We can represent a discrete (i.e., finitely supported) probability distribution over $\R^n$ by its probability mass function $D\from \R^n \to \R$ such that $D \geq 0$ and $\sum_{x \in \mathrm{supp}(D)} D(x) = 1$.
Similarly, we can describe a pseudo-distribution by its mass function by relaxing the constraint $D\ge 0$ to passing certain low-degree non-negativity tests.

Concretely, a \emph{level-$\ell$ pseudo-distribution} is a finitely-supported function $D:\R^n \rightarrow \R$ such that $\sum_{x} D(x) = 1$ and $\sum_{x} D(x) f(x)^2 \geq 0$ for every polynomial $f$ of degree at most $\ell/2$.
(Here, the summations are over the support of $D$.)
A straightforward polynomial-interpolation argument shows that every level-$\infty$-pseudo distribution satisfies $D\ge 0$ and is thus an actual probability distribution.
We define the \emph{pseudo-expectation} of a function $f$ on $\R^d$ with respect to a pseudo-distribution $D$, denoted $\pE_{D(x)} f(x)$, as
\begin{equation}
  \pE_{D(x)} f(x) = \sum_{x} D(x) f(x) \,\mper
\end{equation}
The degree-$\ell$ moment tensor of a pseudo-distribution $D$ is the tensor $\E_{D(x)} (1,x_1, x_2,\ldots, x_n)^{\otimes \ell}$.
In particular, the moment tensor has an entry corresponding to the pseudo-expectation of all monomials of degree at most $\ell$ in $x$.
The set of all degree-$\ell$ moment tensors of probability distribution is a convex set.
Similarly, the set of all degree-$\ell$ moment tensors of degree $d$ pseudo-distributions is also convex.
Unlike moments of distributions, there's an efficient separation oracle for moment tensors of pseudo-distributions.

\begin{fact}[\cite{MR939596-Shor87,parrilo2000structured,MR1748764-Nesterov00,MR1846160-Lasserre01}]
  \label[fact]{fact:sos-separation-efficient}
  For any $n,\ell \in \N$, the following set has a $n^{O(\ell)}$-time weak separation oracle (in the sense of \cite{MR625550-Grotschel81}):
  \begin{equation}
    \Set{ \pE_{D(x)} (1,x_1, x_2, \ldots, x_n)^{\otimes \ell} \mid \text{ degree-$\ell$ pseudo-distribution $D$ over $\R^n$}}\,\mper
  \end{equation}
\end{fact}
This fact, together with the equivalence of weak separation and optimization \cite{MR625550-Grotschel81} allows us to efficiently optimize over pseudo-distributions (approximately)---this algorithm is referred to as the sum-of-squares algorithm. The \emph{level-$\ell$ sum-of-squares algorithm} optimizes over the space of all level-$\ell$ pseudo-distributions that satisfy a given set of polynomial constraints (defined below).

\begin{definition}[Constrained pseudo-distributions]
  Let $D$ be a level-$\ell$ pseudo-distribution over $\R^n$.
  Let $\cA = \{f_1\ge 0, f_2\ge 0, \ldots, f_m\ge 0\}$ be a system of $m$ polynomial inequality constraints.
  We say that \emph{$D$ satisfies the system of constraints $\cA$ at degree $r$}, denoted $D \sdtstile{r}{} \cA$, if for every $S\subseteq[m]$ and for every sum-of-squares polynomial $h$ with $\deg h + \sum_{i\in S} \max\set{\deg f_i,r}$, it holds that $\pE_{D} h \cdot \prod _{i\in S}f_i  \ge 0$.

  We write $D \sdtstile{}{} \cA$ (without specifying the degree) if $D \sdtstile{0}{} \cA$ holds.
  Furthermore, we say that $D\sdtstile{r}{}\cA$ holds \emph{approximately} if the above inequalities are satisfied up to an error of $2^{-n^\ell}\cdot \norm{h}\cdot\prod_{i\in S}\norm{f_i}$, where $\norm{\cdot}$ denotes the Euclidean norm\footnote{The choice of norm is not important here because the factor $2^{-n^\ell}$ swamps the effects of choosing another norm.} of the coefficients of a polynomial in the monomial basis.
\end{definition}

We remark that if $D$ is an actual (discrete) probability distribution, then we have  $D\sdtstile{}{}\cA$ if and only if $D$ is supported on solutions to the constraints $\cA$. We say that a system $\cA$ of polynomial constraints is \emph{explicitly bounded} if it contains a constraint of the form $\{ \|x\|^2 \leq M\}$.
The following fact is a consequence of \cref{fact:sos-separation-efficient} and \cite{MR625550-Grotschel81},

\begin{fact}[Efficient Optimization over Pseudo-distributions]
\label{fact:efficient_sos_opt}
There exists an $(n+ m)^{O(\ell)} $-time algorithm that, given any explicitly bounded and satisfiable system\footnote{Here, we assume that the bit complexity of the constraints in $\cA$ is $(n+m)^{O(1)}$.} $\cA$ of $m$ polynomial constraints in $n$ variables, outputs a level-$\ell$ pseudo-distribution that satisfies $\cA$ approximately. \label{fact:eff-pseudo-distribution}
\end{fact}

\subsection{Sum-of-squares proofs}
\label{subsec:sos_proofs}

Let $f_1, f_2, \ldots, f_r$ and $g$ be multivariate polynomials in $x$.
A \emph{sum-of-squares proof} that the constraints $\{f_1 \geq 0, \ldots, f_m \geq 0\}$ imply the constraint $\{g \geq 0\}$ consists of  polynomials $(p_S)_{S \subseteq [m]}$ such that
\begin{equation}
g = \sum_{S \subseteq [m]} p^2_S \cdot \Pi_{i \in S} f_i
\mper
\end{equation}
We say that this proof has \emph{degree $\ell$} if for every set $S \subseteq [m]$, the polynomial $p^2_S \Pi_{i \in S} f_i$ has degree at most $\ell$.
If there is a degree $\ell$ SoS proof that $\{f_i \geq 0 \mid i \leq r\}$ implies $\{g \geq 0\}$, we write:
\begin{equation}
  \{f_i \geq 0 \mid i \leq r\} \sststile{\ell}{}\{g \geq 0\}
  \mper
\end{equation}
For all polynomials $f,g\colon\R^n \to \R$ and for all functions $F\colon \R^n \to \R^m$, $G\colon \R^n \to \R^k$, $H\colon \R^{p} \to \R^n$ such that each of the coordinates of the outputs are polynomials of the inputs, we have the following inference rules.

The first one derives new inequalities by addition/multiplication:
\begin{equation} \label{eq:sos-addition-multiplication-rule}
\frac{\cA \sststile{\ell}{} \{f \geq 0, g \geq 0 \} } {\cA \sststile{\ell}{} \{f + g \geq 0\}}, \frac{\cA \sststile{\ell}{} \{f \geq 0\}, \cA \sststile{\ell'}{} \{g \geq 0\}} {\cA \sststile{\ell+\ell'}{} \{f \cdot g \geq 0\}}\tag{Addition/Multiplication Rules}\mper
\end{equation}
The next one derives new inequalities by transitivity: 
\begin{equation} \label{eq:sos-transitivity}
\frac{\cA \sststile{\ell}{} \cB, \cB \sststile{\ell'}{} C}{\cA \sststile{\ell \cdot \ell'}{} C}\tag{Transitivity Rule}\mcom 
\end{equation}
Finally, the last rule derives new inequalities via substitution:
\begin{equation} \label{eq:sos-substitution}
\frac{\{F \geq 0\} \sststile{\ell}{} \{G \geq 0\}}{\{F(H) \geq 0\} \sststile{\ell \cdot \deg(H)} {} \{G(H) \geq 0\}}\tag{Substitution Rule}\mper
\end{equation}

Low-degree sum-of-squares proofs are sound and complete if we take low-level pseudo-distributions as models.
Concretely, sum-of-squares proofs allow us to deduce properties of pseudo-distributions that satisfy some constraints.
\begin{fact}[Soundness]
  \label{fact:sos-soundness}
  If $D \sdtstile{r}{} \cA$ for a level-$\ell$ pseudo-distribution $D$ and there exists a sum-of-squares proof $\cA \sststile{r'}{} \cB$, then $D \sdtstile{r\cdot r'+r'}{} \cB$.
\end{fact}
If the pseudo-distribution $D$ satisfies $\cA$ only approximately, soundness continues to hold if we require an upper bound on the bit-complexity of the sum-of-squares $\cA \sststile{r'}{} B$  (number of bits required to write down the proof). In our applications, the bit complexity of all sum of squares proofs will be $n^{O(\ell)}$ (assuming that all numbers in the input have bit complexity $n^{O(1)}$). This bound suffices in order to argue about pseudo-distributions that satisfy polynomial constraints approximately.

The following fact shows that every property of low-level pseudo-distributions can be derived by low-degree sum-of-squares proofs.
\begin{fact}[Completeness]
  \label{fact:sos-completeness}
  Suppose $d \geq r' \geq r$ and $\cA$ is a collection of polynomial constraints with degree at most $r$, and $\cA \vdash \{ \sum_{i = 1}^n x_i^2 \leq B\}$ for some finite $B$.

  Let $\{g \geq 0 \}$ be a polynomial constraint.
  If every degree-$d$ pseudo-distribution that satisfies $D \sdtstile{r}{} \cA$ also satisfies $D \sdtstile{r'}{} \{g \geq 0 \}$, then for every $\epsilon > 0$, there is a sum-of-squares proof $\cA \sststile{d}{} \{g \geq - \epsilon \}$.
\end{fact}
We will use the following Cauchy-Schwarz inequality for pseudo-distributions:
\begin{fact}[Cauchy-Schwarz for Pseudo-distributions]
Let $f,g$ be polynomials of degree at most $d$ in indeterminate $x \in \R^d$. Then, for any degree d pseudo-distribution $\tmu$,
$\pE_{\tmu}[fg] \leq \sqrt{\pE_{\tmu}[f^2]} \sqrt{\pE_{\tmu}[g^2]}$.
 \label{fact:pseudo-expectation-cauchy-schwarz}
\end{fact}

\begin{fact}[Hölder's Inequality for Pseudo-Distributions] \label{fact:pseudo-expectation-holder}
Let $f,g$ be polynomials of degree at most $d$ in indeterminate $x \in \R^d$. 
Fix $t \in \N$. Then, for any degree $dt$ pseudo-distribution $\tmu$,
$\pE_{\tmu}[f^{t-1}g] \leq \paren{\pE_{\tmu}[f^t]}^{\frac{t-1}{t}} \paren{\pE_{\tmu}[g^t]}^{1/t}$.
\end{fact}

\begin{fact}[Almost Triangle Inequality] \label{fact:sos-almost-triangle}
Let $a,b$ be indeterminates. Then, for any $t \in \N$, 
\[
\sststile{2t}{a,b} \Set{(a+b)^{2t} \leq 2^{2t} \Paren{a^{2t} + b^{2t}}}\mper
\]
\end{fact}
\begin{fact}[Operator norm Bound]
\label{fact:operator_norm}
Let $A$ be a symmetric $d\times d$ matrix and $v$ be a vector in $\mathbb{R}^d$. Then,
\[
\sststile{2}{v} \Set{ v^{\top} A v \leq \|A\|_2\|v\|^2_2 }
\]
\end{fact}

\begin{fact}[SoS AM-GM Inequality, see Appendix A of~\cite{MR3388192-Barak15}] \label{fact:sos-am-gm}
Let $f_1, f_2,\ldots, f_m$ be indeterminates. Then, 
\[
\sststile{m}{f_1, f_2,\ldots, f_m} \Set{ \Paren{\frac{1}{m} \sum_{i =1}^n f_i }^m \geq \Pi_{i \leq m} f_i} \mper
\]
\end{fact}

\begin{fact}[SoS Hölder's Inequality]\label{fact:sos-holder}
Let $w_1, \ldots w_n$ be indeterminates and let $f_1,\ldots f_n$ be polynomials of degree $m$ in vector valued variable $x$. 
Let $k$ be a power of 2.  
Then, 
\[
\Set{w_i^2 = w_i, \forall i\in[n] } \sststile{2km}{x,w} \Set{  \Paren{\frac{1}{n} \sum_{i = 1}^n w_i f_i}^{k} \leq \Paren{\frac{1}{n} \sum_{i = 1}^n w_i}^{k-1} \Paren{\frac{1}{n} \sum_{i = 1}^n f_i^k}} 
\]
\end{fact}

\begin{lemma}[Cancellation within SoS, similar to Lemma~9.2 in~\cite{bakshi2020outlierrobust}]\label{fact:cancellation-within-SoS}
Let $a$ be an indeterminate. Then,
\[
\Set{a^t \leq 1} \cup \Set{a \geq 0} \sststile{t}{a} \Set{a \leq 1}
\]
\end{lemma}

\begin{proof}
Applying the SoS AM-GM inequality (Fact~\ref{fact:sos-am-gm}) with $f_1 = a$, $f_2 = \ldots = f_t = 1$, we get:
\[
\sststile{t}{a} \Set{a \leq a^{t}/t + 1-1/t} \mper
\]
Thus, 
\[
\Set{a^{t} \leq 1} \sststile{t}{a} \Set{a \leq 1/t + 1-1/t = 1}\mper
\]
\end{proof}

\begin{fact}[Cancellation within SoS, Lemma~9.3 in~\cite{bakshi2020outlierrobust}] \label{fact:cancellation-within-SoS-indeterminate}
Let $a,C$ be indeterminates. Then, 
\[
\Set{a \geq 0} \cup \Set{ a^t \leq Ca^{t-1}} \sststile{2t}{a,C} \Set{a^{2t} \leq C^{2t}}\mper
\]
\end{fact}

The following fact is a simple corollary of the fundamental theorem of algebra:
\begin{fact}
For any univariate degree $d$ polynomial $p(x) \geq 0$ for all $x \in \R$,
$\sststile{d}{x} \Set{p(x) \geq 0}$.
 \label{fact:univariate}
\end{fact}
This can be extended to univariate polynomial inequalities over intervals of $\R$.
2
\begin{fact}[Fekete and Markov-Lukacs, see \cite{MR2500468-Laurent09}]
\label{fact:fekete_univariate}
For any univariate degree $d$ polynomial $p(x) \geq 0$ for $x \in [a, b]$,  $\Set{x\geq a, x \leq b} \sststile{d}{x} \Set{p(x) \geq 0}$.  \label{fact:univariate-interval}
\end{fact}

\begin{fact}\label{fact:spectral-sos}
Let $A \succeq 0$ be a $d \times d$ matrix. Then,
\[
\sststile{2}{v} \Set{ v^{\top} A v \geq 0}\mper
\]
\end{fact}

\paragraph{Re-weightings Pseudo-distributions.}
The following fact is easy to verify and has been used in several works (see~\cite{DBLP:conf/stoc/BarakKS17} for example).  
\begin{fact}[Re-weighting] \label{fact:reweightings}
Let $\tmu$ be a pseudo-distribution of degree $k$ satisfying a set of polynomial constraints $\cA$ in variable $x$. 
Let $p$ be a sum-of-squares polynomial of degree $t$ such that $\pE[p(x)] \neq 0$.
Let $\tmu'$ be the pseudo-distribution defined so that for any polynomial $f$, $\pE_{\tmu'}[f(x)] = \pE_{\tmu}[ f(x)p(x)]/\pE_{\tmu}[p(x)]$. Then, $\tmu'$ is a pseudo-distribution of degree $k-t$ satisfying $\cA$. 
\end{fact}

\section{List-decodable Subspace Recovery in Fixed Polynomial Time}
In this section, we prove Theorem~\ref{thm:main} restated below.
\begin{theorem}[List-decodable subspace recovery] \label{thm:main-section}
Let $\Pi_*$ be a projector a subspace of dimension $r$. 
Let $\cD$ be a distribution with certifiably $C$-hypercontractive degree 2 polynomials. 
Then, Algorithm \ref{alg:subspace-recovery} takes as input $n = n_0 \geq (d/\alpha)^{16}$ samples from $\sub_{\cD}(\alpha, \Pi_*)$ and outputs a list $\cL$ of $O(1/\alpha)$  projection matrices such that with probability at least $0.99$ over the draw of the sample and the randomness of the algorithm, there is a $\hat{\Pi} \in \cL$ satisfying $\|\hat{\Pi} - \Pi_*\|_F \leq O(C^2/\alpha)$. Further, Algorithm \ref{alg:subspace-recovery} has time complexity at most $O(n^{18})$. 
\end{theorem}

Let $\cA_{w,\Pi}$ be the following system of polynomial inequality constraints in indeterminates $w,\Pi$: 
\begin{equation}
  \cA_{w,\Pi}\colon
  \left \{
    \begin{aligned}
      &&
      \textstyle\sum_{i\in[n]} w_i
      &= \alpha n\\
      &\forall i\in [n].
      & w_i \Pi x_i
      & = w_i x_i \\
      &\forall i\in [n].
      & w_i^2
      & = w_i \\
      &
      & \Pi^2
      & = \Pi\\
      &\forall Q 
      &\frac{1}{\alpha n} \sum_{i \leq n} w_i \Paren{x_i^{\top} Q x_i-\tr(Q\Pi)}^{2t} 
      &\leq (Ct)^t \Paren{\frac{1}{\alpha n} \sum_{i \leq n} w_i \Paren{x_i^{\top} Q x_i}^2}^t\\
      &\forall Q 
      &\frac{1}{\alpha n} \sum_{i \leq n} w_i \Paren{x_i^{\top} Q x_i-\tr(Q\Pi)}^{2} 
      &\leq C \Norm{\Pi Q\Pi}_F^2\mper
    \end{aligned}
  \right \}
\end{equation}

Here, $w$ indicates a subsample of the input $x$ and $\Pi$ is the target projector matrix.
In order to force natural properties of being a projector matrix, we add the constraint that $\Pi^2 = \Pi$ for the matrix valued indeterminate $\Pi$. The constraints on $w$ encode that the subsample indicated by $w$ satisfies 1) $\Pi x_i = x_i$ for every $i$ corresponding to $w_i = 1$, that 2) the uniform distribution on the subsample has certifiably hypercontractive degree 2 polynomials, and that 3) the uniform distribution on the subsample has degree $2$ polynomials of bounded variance.

Note that the certifiable hypercontractivity constraints can be encoded as a system of $d^{O(t)}$ inequality constraints by using the standard technique of constraint compression (first used in~\cite{DBLP:journals/corr/abs-1711-11581,HopkinsLi17}) in SoS proofs (see~\cite{TCS-086}, Chapter 4.3 for an exposition).


Our algorithm simply finds a high-entropy pseudo-distribution consistent with $\cA_{w,\Pi}$ and uses the ``rounding by votes'' scheme introduced by Karmalkar, Klivans and Kothari ~\cite{DBLP:journals/corr/abs-1905-05679} to complete the proof.

\begin{mdframed}
  \begin{algorithm}[List-Decodable Subspace Recovery]
    \label[algorithm]{alg:subspace-recovery}\mbox{}
    \begin{description}
    \item[Given:]
    Sample $Y = \{x_1, x_2, \ldots x_n \}$  drawn according to $\sub_{\cD}(\alpha,\Pi_*)$ with inliers $\cI$.
    \item[Output:]
      A list $L$ of $O(1/\alpha)$ projection matrices such that there exists $\hat{\Pi} \in L$ satisfying $\|\hat{\Pi} -\Pi_*\|_F < O(1/\alpha)$.
    \item[Operation:]\mbox{}
    \begin{enumerate}
    \item Find a degree-$18$ pseudo-distribution $\tilde{\mu}$ satisfying $\cA_{w,\Pi}$ that minimizes $\|\pE[w]\|_2$.
    \item For each $i \in [n]$ such that $\pE_{\tmu}[w_i] > 0$, let $\hat{\Pi}_i = \frac{\pE_{\tmu}[w_i \Pi]}{\pE_{\tmu}[w_i]}$. Otherwise, set $\hat{\Pi}_i =0$.
    \item 
    Take $J$ be a random multi-set formed by union of $O(1/\alpha)$ independent draws of $i \in [n]$ with probability $\frac{\pE[w_i]}{\alpha n}$.
    \item Output $L = \{\hat{\Pi}_i \mid i \in J\}$ where $J \subseteq [n]$.
  \end{enumerate}
    \end{description}    
  \end{algorithm}
\end{mdframed}

\subsection{Analysis of Algorithm \ref{alg:subspace-recovery}.}

The main new component of our analysis is a polynomial inequality that directly relates the Frobenius distance between $\Pi$, the program variable, and the true projection matrix $\Pi_*$. Further, the inequality admits a sum-of-squares proof. 
\begin{lemma}[Frobenius Closeness of $\Pi$ and $\Pi_*$] 
\label{lem:close-projectors}
For every $t \in \N$,
\begin{align*}
\cA_{w,\Pi} \sststile{16t^2+2t}{w,\Pi} \Biggl\{ w(\cI) \Norm{\Pi-\Pi_*}_F^{2t} \leq (4C^2t)^{2t}\Biggr\}
\end{align*}
where $w(\cI) = \frac{1 }{|\cI|} \sum_{i \in \cI} w_i$. 
\end{lemma}

As in the prior works on list-decodable learning, we will combine the lemma above with the following consequence of optimizing the high-entropy objective on $\tmu$:

\begin{fact}[Large weight on inliers from high-entropy constraints, same proof as Lemma 4.3 in~\cite{DBLP:journals/corr/abs-1905-05679}] \label{fact:high-entropy-pseudo-distributions}
Let $\tmu$ pseudo-distribution of degree $\geq 2$ that satisfies $\cA_{w,\Pi}$ and minimizes $\Norm{\pE_{\tmu} \sum_{i \in [n]} w_i}_2$. Then, $\pE_{\tmu} [w(\cI)] \geq \alpha$.
\end{fact}

It's easy to finish the proof of Theorem~\ref{thm:main} using the above two claims:

\begin{proof}[Proof of Main Theorem~\ref{thm:main}]
First, since $\cD$ is certifiably $C$-hypercontractive, Lemma~\ref{lem:cert_hyper_sampling} implies that $\geq n = \Omega(d\log(d)/\alpha)^{16}$ samples suffice for the uniform distribution on the inliers, $\cI$, to have $2$-certifiably $C$-hypercontractive degree 2 polynomials with probability at least $1-1/d$. 
Let $\zeta_1$ be the event that this succeeds, and condition on it. 

Let $\tmu$ be a pseudo-distribution of degree-$24$ satisfying $\cA_{w,\Pi}$ and minimizing $\Norm{\pE[w]}_2$ as described in Algorithm \ref{alg:subspace-recovery}. Observe, such a pseudo-distribution is guaranteed to exist: take the pseudo-distribution supported on a single point, $(w,\Pi)$ such that $w_i = 1$ iff $i \in \cI$ and $\Pi = \Pi_*$. It is straight forward to check that $w_i(\Id - \Pi_*)x_i = 0$, $\Pi_*$ is indeed a rank $k$ projection matrix and $\sum_{i \in [n]} w_i = \alpha n$. Conditioned on $\zeta_1$, the hypercontractivity constraint is also satisfied by the inliers.  

Since Lemma~\ref{lem:close-projectors} admits a sum-of-squares proof, it follows from Fact \ref{fact:sos-soundness}, that the polynomial inequality is preserved under pseudo-expectations. Instantiating with $t=1$, we have,

\begin{equation*}
\frac{1}{|\cI|} \sum_{i \in \cI} \pE_{\tmu} \left[w_i^2 \cdot \Norm{\Pi -\Pi_*}_F^{2} \right] \leq  (4C^2)^{2}\mper
\end{equation*}
Alternatively, we can rewrite the above as follows: 
\begin{equation*}
\frac{1}{|\cI|} \sum_{i \in \cI} \pE_{\tmu}\left[ \Norm{w_i\Pi -w_i\Pi_*}_F^{2}\right]  \leq  (4C^2)^{2}\mper
\end{equation*} 
Applying Jensen's inequality yields, 

\begin{equation*}
\Paren{\frac{1}{|\cI|} \sum_{i \in \cI} \Norm{\pE_{\tmu}[w_i\Pi] -\pE_{\tmu}[w_i]\Pi_*}_F}^{2}  \leq  (4C^2)^{2}\mper
\end{equation*} 
Taking the square-root, 

\begin{equation*}
\Paren{\frac{1}{|\cI|} \sum_{i \in \cI} \Norm{\pE_{\tmu}[w_i\Pi] -\pE_{\tmu}[w_i]\Pi_*}_F}  \leq  4C^2\mper
\end{equation*}
Recall, the rounding in Algorithm \ref{alg:subspace-recovery} uses $\hat{\Pi}_i = \pE[w_i \Pi]/\pE[w_i]$ to denote the projector corresponding to the $i$-th sample. Then, rewriting the above equation yields:
\begin{equation*}
\Paren{\frac{1}{|\cI|} \sum_{i \in \cI} \pE[w_i] \cdot \|\hat{\Pi}_i -\Pi_*\|_F}  \leq 4C^2\mper
\end{equation*}         

Let $Z = \frac{1}{\alpha n} \sum_{i \in \cI} \pE[w_i]$. Then, from Fact~\ref{fact:high-entropy-pseudo-distributions}, $Z \geq \alpha$. Dividing by $Z$ on both sides thus yields:

\begin{equation}
\frac{1}{Z} \Paren{\frac{1}{|\cI|} \sum_{i \in \cI} \pE[w_i] \cdot \|\hat{\Pi}_i -\Pi_*\|_F}  \leq  4C^2/\alpha\mper \label{eq:good-on-average}
\end{equation} 

Since each index $i \in [n]$ is chosen with probability $\frac{\pE[w_i]}{\sum_{i \in [n]} \pE[w_i]} = \frac{1}{\alpha n} \pE[w_i]$, it follows that $i \in \cI$ with probability at least $\frac{1}{\alpha n}\sum_{i \in \cI} \pE[w_i] = Z \geq \alpha$.  
By Markov's inequality applied to \eqref{eq:good-on-average},  with probability $\frac{1}{2}$ over the choice of $i$ conditioned on $i \in \cI$, $\|\hat{\Pi}_i - \Pi_*\|_2 \leq 8C^2/\alpha$. Thus, in total, with probability at least $\alpha/2$, $\|\hat{\Pi}_i - \Pi_*\|_F \leq 8C^2/\alpha$. Thus, the with probability at least $0.99$ over the draw of the random set $J$, the list constructed by the algorithm contains  $\hat{\Pi}$ such that $\|\hat{\Pi} - \Pi_*\|_2 \leq 8C^2/\alpha$.

Let us now account for the running time and sample complexity of the algorithm.
The sample size for the algorithm is dictated by Lemma~\ref{lem:cert_hyper_sampling} and thus requires $n_0=\Omega\left((d\log(d)/\alpha)^{16}\right)$.
A degree-$18$ pseudo-distribution satisfying $\cA_{w,\ell}$ and minimizing $\|\pE[w]\|_2$ can be found in time $O(n^{18})$ and dominates the running time. 
The rounding procedure runs in time at most $O(nd)= O(n^2)$. 
\end{proof}

\subsection{Proof of Lemma~\ref{lem:close-projectors}}




\begin{lemma}[Frobenius Closeness of Subsample to Covariance, $w$-Samples] \label{lem:frobenius-closeness-subsample-w}
\begin{equation}
\cA_{w,\Pi} \sststile{w,\Pi}{4t} \Biggl\{ \Norm{\frac{1}{\alpha n} \sum_{i \in \cI} w_i x_i x_i ^{\top}- w(\cI) \Pi}_F^{4t} \leq w(\cI)^{4t-2} (C^2 t)^{4t} \Biggr\}
\end{equation}
\end{lemma}

\begin{proof}

For a $d \times d$ matrix-valued indeterminate $Q$, using the SoS Hölder's Inequality, we have
\begin{equation}
\label{eqn:sos_holder_decouple}
\begin{split}
\cA_{w,\Pi} \sststile{4t}{w, \Pi, Q} \Biggl\{ \Iprod{ \frac{1}{\alpha n} \sum_{i \in \cI} w_i x_i x_i ^{\top}- w(\cI) \Pi, Q}^{2t} &= \Iprod{ \frac{1}{\alpha n} \sum_{i \in \cI} w_i^{2t} \Paren{x_i x_i ^{\top}- \Pi}, Q}^{2t} \\
&\leq \Paren{\frac{1}{\alpha n} \sum_{i \in \cI} w_i^{2t}}^{2t-1}  \frac{1}{\alpha n} \sum_{i\in\cI} w_i \Iprod{ x_i x_i ^{\top}-  \Pi, Q}^{2t}\\
&\leq \Paren{\frac{1}{\alpha n} \sum_{i \in \cI} w_i^{2t}}^{2t-1}  \frac{1}{\alpha n} \sum_{i\leq n} w_i \Iprod{ x_i x_i ^{\top}-  \Pi, Q}^{2t}  \Biggr\} \\
\end{split}
\end{equation}

Using certifiable hypercontractivity of $w$-samples, combined with the bounded variance constraints, we have:

\begin{equation}
\label{eqn:hypercontractivity_variance}
\begin{split}
\cA_{w,\Pi} \sststile{4t}{w, \Pi, Q} \Biggl\{ \frac{1}{\alpha n} \sum_{i\leq n} w_i \Iprod{x_i x_i ^{\top}- \Pi, Q}^{2t} & \leq (Ct)^t \Paren{\frac{1}{\alpha n} \sum_{i\leq n} w_i \Iprod{x_i x_i ^{\top}- \Pi, Q}^{2}}^t\\
&\leq  (C^2t)^{2t} \Norm{\Pi Q \Pi}_F^{2t} \hspace{0.2in} \Biggr\} \\
\end{split}
\end{equation}
Combining Equations \eqref{eqn:sos_holder_decouple} and \eqref{eqn:hypercontractivity_variance}, and substituting $Q = \frac{1}{\alpha n} \sum_{i \in \cI} w_i x_i x_i ^{\top}- w(\cI) \Pi$, we have 



\begin{equation}
\cA_{w,\Pi} \sststile{4t}{w,\Pi} \Biggl\{ \Norm{\frac{1}{\alpha n} \sum_{i \in \cI} w_i x_i x_i ^{\top}- w(\cI) \Pi}_F^{4t} \leq w(\cI)^{2t-1} (C^2 t)^{2t}  \Norm{\frac{1}{\alpha n} \sum_{i \in \cI} w_i x_i x_i ^{\top}- w(\cI) \Pi}_F^{2t} \Biggr\}
\end{equation}
where we used that $\cA_{w,\Pi} \sststile{}{}\Set{w_i \Pi x_i = w_i x_i}$ and $\cA_{w,\Pi} \sststile{}{}\Set{\Pi^2 = \Pi}$.

We now apply cancellation within SoS (Fact~\ref{fact:cancellation-within-SoS-indeterminate}), with $a = \Norm{\frac{1}{\alpha n} \sum_{i \in \cI} w_i x_i x_i ^{\top}- w(\cI) \Pi}_F^{2t}$ and $C = w(\cI)^{2t-1} (C^2 t)^{2t}$ gives:

\begin{equation}
\cA_{w,\Pi} \sststile{4t}{w,\Pi} \Biggl\{ \Norm{\frac{1}{\alpha n} \sum_{i \in \cI} w_i x_i x_i ^{\top}- w(\cI) \Pi}_F^{4t} \leq w(\cI)^{4t-2} (C^2 t)^{4t} \Biggr\}
\end{equation}
This completes the proof. 

\end{proof}

By a similar argument, we obtain:
\begin{lemma}[Frobenius Closeness of Subsample to Covariance, $\cI$-Samples] \label{lem:frobenius-closeness-subsample-inliers}
\begin{equation}
\cA_{w,\Pi} \sststile{4t}{w,\Pi} \Biggl\{ \Norm{\frac{1}{\alpha n} \sum_{i \in \cI} w_i x_i x_i ^{\top}- w(\cI) \Pi_*}_F^{4t} \leq w(\cI)^{4t-2} (C^2 t)^{4t} \Biggr\}
\end{equation}
\end{lemma}

We can now simply combine Lemmas~\ref{lem:frobenius-closeness-subsample-w} and ~\ref{lem:frobenius-closeness-subsample-inliers} to prove Lemma~\ref{lem:close-projectors}.

\begin{proof}[Proof of Lemma~\ref{lem:close-projectors}]
Using the SoS Almost triangle inequality (Fact~\ref{fact:sos-almost-triangle}), and Lemmas~\ref{lem:frobenius-closeness-subsample-w} and ~\ref{lem:frobenius-closeness-subsample-inliers}, we have:
\begin{align*}
\cA_{w,\Pi} \sststile{4t}{w,\Pi} \Biggl\{ w(\cI)^{4t} \Norm{\Pi-\Pi_*}_F^{4t} &\leq 2^{4t} \Norm{\frac{1}{\alpha n} \sum_{i \in \cI} w_i x_i x_i ^{\top}- w(\cI) \Pi_*}_F^{4t} + 2^{4t} \Norm{\frac{1}{\alpha n} \sum_{i \in \cI} w_i x_i x_i ^{\top}- w(\cI) \Pi}_F^{4t}\\
&\leq 2^{4t+1} w(\cI)^{4t-2} (C^2 t)^{4t} \\
&\leq w(\cI)^{4t-2} (4C^2t)^{4t} \hspace{0.1in}\Biggr\}
\end{align*}
Multiplying both sides of the inequality above by $\Norm{\Pi-\Pi_*}_F^{8t^2-4t}$, and using the \ref{eq:sos-substitution} we obtain:
\begin{align*}
\cA_{w,\Pi} \sststile{16t^2}{w, \Pi} \Biggl\{ w(\cI)^{4t} \Norm{\Pi-\Pi_*}_F^{8t^2} \leq w(\cI)^{4t-2} \Norm{\Pi-\Pi_*}_F^{8t^2-4t} (4C^2t)^{4t}\Biggr\}
\end{align*}
Applying SoS Cancellation (Fact~\ref{fact:cancellation-within-SoS-indeterminate}) with $a \rightarrow w(\cI)^2 \Norm{\Pi-\Pi_*}_F^{4t}$ and $C \rightarrow (4C^2t)^{4t}$, we obtain:

\begin{align*}
\cA_{w,\Pi} \sststile{16t^2}{w, \Pi} \Biggl\{ w(\cI)^{8t} \Norm{\Pi-\Pi_*}_F^{16t^2} \leq (4C^2t)^{16t^2}\Biggr\}
\end{align*}
Applying Fact~\ref{fact:cancellation-within-SoS} with $a = w(\cI)^{} \Norm{\Pi-\Pi_*}_F^{2t} (4C^2t)^{-2t}$, $t= 8t$, and rearranging, we finally have:
\begin{align*}
\cA_{w,\Pi} \sststile{16t^2}{w, \Pi} \Biggl\{ w(\cI) \Norm{\Pi-\Pi_*}_F^{2t} \leq (4C^2t)^{2t}\Biggr\}
\end{align*}

\end{proof}


\section{High-Precision List-Decodable Subspace Recovery }
\label{sec:high_precision_LDS}

In this section, we describe an efficient algorithm for list-decodable subspace recovery with arbitrarily small error at the cost of returning a polynomial (in the dimension $d$) size list.

\begin{theorem}[Large-List Subspace Recovery, Theorem \ref{thm:main-2} restated]\label{thm:large-list-rounding}
Let $\cD$ be $k$-certifiably $(c,\alpha/2C)$-anti-concentrated distribution with certifiably $C$-hypercontractive degree 2 polynomials.
Then, for any $\eta >0$, there exists an algorithm that takes input $n \geq n_0 = (kd\log(d)/\alpha)^{O(k)}$ samples from $\sub_{\cD}(\alpha, \Pi_*)$ and outputs a list $\cL$ of size $O(1/\alpha^{\Paren{\log k + \log(1/\eta)}})$ of projection matrices such that with probability at least $0.9$ over the draw of the sample and the randomness of the algorithm, there is a $\hat{\Pi} \in \cL$ satisfying $\|\hat{\Pi} - \Pi_*\|_F \leq \eta$. The algorithm has time complexity at most $n^{O(k + \log(1/\eta))}$. 
\end{theorem}

We again use the system of polynomial constraints, $\cA_{w,\Pi}$, in indeterminates $w,\Pi$, and include a certifiable anti-concentration constraint on the samples indicated by $w$.
\begin{equation}
  \cA_{w,\Pi}\colon
  \left \{
    \begin{aligned}
      &&
      \textstyle\sum_{i\in[n]} w_i
      &= \alpha n\\
      &\forall i\in [n].
      & w_i \Pi x_i
      & = w_i x_i \\
      &\forall i\in [n].
      & w_i^2
      & = w_i \\
      &
      & \Pi^2
      & = \Pi\\
      &
      &\Tr(\Pi)
      &= r\\
      &\forall Q 
      &\frac{1}{\alpha n} \sum_{i \leq n} w_i \Paren{x_i^{\top} Q x_i-\tr(Q\Pi)}^{2t} 
      &\leq (Ct)^t \Paren{\frac{1}{\alpha n} \sum_{i \leq n} w_i \Paren{x_i^{\top} Q x_i}^2}^t\\
      &\forall Q 
      &\frac{1}{\alpha n} \sum_{i \leq n} w_i \Paren{x_i^{\top} Q x_i-\tr(Q\Pi)}^{2} 
      &\leq C \Norm{\Pi Q\Pi}_F^2\mper%
    \end{aligned}
  \right \}
\end{equation}

\begin{mdframed}[nobreak=true]
  \begin{algorithm} List-Decodable Subspace Recovery with Arbitrarily Tiny Error
    \label[algorithm]{alg:subspace-recovery-tiny-error}\mbox{}
    \begin{description}
    \item[Given:]
    Sample $X = \{x_1, x_2, \ldots x_n \} = \cI \cup \cO$ of size $n$ drawn according to $\sub_{\cD}(\alpha,\Pi_*)$ such that $\cD$ is $k$-certifiably $(c,\delta)$-anti-concentrated and $k$-certifiably $c$-hypercontractive. 
      \item[Operation:] \textrm{ }\\
      \vspace{-0.3in}
      \begin{enumerate}
        \item Let $t = \Theta(\log k + \log (1/\eta))$. Find a degree $O(2k + t)$ pseudo-distribution $\tilde{\mu}$ satisfying $\cA_{w,\Pi}$ that minimizes $\|\pE_{\tmu}[w]\|_2$.
        \item  For every multi-set $S \subseteq [n]$ of size $t$ such that $\pE_{\tmu}[w_S] > 0$, let $\hat{\Pi}_S = \frac{\pE_{\tmu}[w_S \Pi]}{\pE_{\tmu}[w_S]}$. Otherwise, set $\hat{\Pi}_S =0$
        \item Take $J$ be a random multi-set formed by union of $O(1/\alpha^t)$ independent draws of $S \subseteq [n]$ of size $t$ with probability proportional to $\binom{I}{ S} \pE[w_S]$.
        \item Output $\cL = \{\hat{\Pi}_S \mid S \in J\}$.
      \end{enumerate}
    \item[Output:] A list $\cL$ of size $O(1/\alpha^{O(\log (1/\alpha) + \log (1/\eta))})$ containing $\hat{\Pi}$ such that $\Norm{\hat{\Pi}-\Pi_*}_F \leq \eta$.   
     \end{description}
  \end{algorithm}
\end{mdframed}

\subsection{Analysis of Algorithm \ref{alg:subspace-recovery-tiny-error}.}
The key technical piece in our analysis is the following polynomial inequality that can be derived from the constraint system $\cA_{w,\Pi}$ in the low-degree SoS proof system.

\begin{lemma} \label{lem:main-sosized-fact}
Fix $\delta >0$ and $t \in \N$. Let $X$ be a sample from $\sub_{\cD}(\alpha, \Pi_*)$ such that $\cI$ is $k$-certifiably $(C,\delta)$-anti-concentrated. Then, 
\[
\cA_{w,\Pi} \sststile{2k + t}{\Pi,w} \Set{ w(\cI)^{t+1} \Norm{\Pi-\Pi_*}_F^k  \leq (C\delta)^t (Ck)^k (4C^2k)^k} \mper
\]
\end{lemma}

Similar to the previous section, we rely on the following consequence of high-entropy constraints on the the pseudo-distribution $\tmu$.
\begin{fact}[Large weight on inliers from high-entropy constraints] \label{fact:high-entropy-pseudo-distributions-2}
Let $\tmu$ pseudo-distribution of degree $\geq t$ that satisfies $\cA_{w,\Pi}$ and minimizes $\Norm{\pE_{\tmu'} \sum_{i \in [n]} w_i}_2$. Then, $\frac{1}{|\cI|^t} \pE \left[\Paren{\sum_{i \in \cI}  w_i }^t\right] \geq \alpha^t$.
\end{fact}

The above two lemmas  combined with the analysis of our rounding allow us to complete the proof of correctness of Algorithm~\ref{alg:subspace-recovery-tiny-error}.

\subsection{Rounding Pseudo-distributions to a Large List: Proof of Theorem ~\ref{thm:large-list-rounding}}

Next, we analyze the new rounding scheme we present in this work that trades off accuracy and list size. The key lemma in our rounding is as follows:

\begin{lemma} \label{lem:step-1-large-rounding}
Given $t \in \N$, and an instance of $\sub_{\cD}(\alpha, \Sigma_*)$ such that $\cI$ is $k$-certifiably $(C, \delta)$-anti-concentrated, let $\tmu$ be a degree-$\Omega(k^2+t)$ pseudo-distribution satisfying $\cA_{w,\Pi}$ and minimizing $\norm{\pE_{\tmu}{[w]}}_2$. Then, for $t = \Omega(\log k + \log 1/\eta)$ and $\delta = \frac{\alpha}{2C}$, we have: 
\[
 \frac{ 1 }{\pE_{\tilde{\mu}} \left[\Paren{\sum_{i \in \cI} w_i}^t\right]} \sum_{S \subseteq \cI, |S|\leq t} \binom{\cI}{S} \pE_{\tilde{\mu}} \left[w_S \Norm{\Pi-\Pi_*}_F^k\right] \leq   \eta/100  \mper
\]
where $\binom{\cI}{S}$ is the coefficient of the monomial indexed by $S$.
\end{lemma}

\begin{proof}
From Lemma~\ref{lem:main-sosized-fact}, we have for every $t,k  \in \N$,

\begin{equation*}
\cA_{w,\Pi} \sststile{O(k^2) + tk}{\Pi,w} \Set{ w(\cI)^{tk} \Norm{\Pi-\Pi_*}_F^k  \leq (C\delta)^{tk-1} (Ck)^k (4C^2k)^k} \mper
\end{equation*}
Since $\tmu$ satisfies $\cA_{w,\Pi}$ and has degree $\Omega(k^2+tk)$, taking pseudo-expectation yields:

\begin{equation}
\pE_{\tmu} \left[ w(\cI)^{tk}  \Norm{\Pi-\Pi_*}_F^k \right] \leq (C\delta)^{tk-1} (Ck)^k (4C^2k)^k \mper
\end{equation}
By Hölder's inequality for pseudo-distributions (Fact~\ref{fact:pseudo-expectation-holder}), we have:

\[
\pE_{\tmu} \left[ w(\cI)^{2t}  \Norm{\Pi-\Pi_*}_F^2 \right] \leq (C\delta)^{2t-2/k} (Ck)^2 (4C^2k)^2 \mper
\]
Since $\tmu$ satisfies $\cA_{w,\Pi}$ and minimizes $\Norm{\pE_{\tilde{\mu}} w}_2$, Fact~\ref{fact:high-entropy-pseudo-distributions} yields: $ \pE_{\tilde{\mu}} \left[ w(\cI)^t \right] \geq \alpha^t$.
Multiplying both sides by $\frac{1}{\pE_{\tilde{\mu}} \left[w(\cI)^t\right]}\leq \frac{1}{\alpha^t}$, we obtain:

\begin{equation}\label{eq:normalized-bound-prior}
 \frac{ 1 }{\pE_{\tilde{\mu}} \left[\Paren{\sum_{i \in \cI} w_i}^t\right]} \pE_{\tmu}\left[\Paren{\sum_{i \in \cI} w_i}^t  \Norm{\Pi-\Pi_*}_F^2\right] \leq   (C\delta)^{2t-2/k} (Ck)^2 (4C^2k)^2 \alpha^{-t} \mper
\end{equation}
Choosing $\delta = \frac{\alpha}{2C}$ and $t = 10 C \Paren{\log(k) + \log(100/\eta)}$ yields:

\begin{equation}\label{eq:normalized-bound}
 \frac{ 1 }{\pE_{\tilde{\mu}} \left[\Paren{\sum_{i \in \cI} w_i}^t\right]} \pE_{\tmu}\left[\Paren{\sum_{i \in \cI} w_i}^t  \Norm{\Pi-\Pi_*}_F^2\right] \leq   \eta/100 \mper
\end{equation}


For any monomial $w_S$, let $w_{S'}$ be its multi-linearization. Then, observe that:
\[
\Set{w_i^2=w_i \mid \forall i} \sststile{t}{w} \Set{ w_S = w_{S'}}\mper
\]
Therefore, we have
\begin{equation}
\label{eqn:expanding_to_monomial}
\cA_{w,\Pi} \sststile{t}{w} \Set{ \Paren{\sum_{i \in \cI} w_i}^t \Norm{\Pi-\Pi_*}_F^2 = \sum_{S \subseteq \cI, |S| \leq t} \binom{\cI}{S} w_S \Norm{\Pi-\Pi_*}_F^2} \mper
\end{equation}
Taking pseudo-expectations for Equation \eqref{eqn:expanding_to_monomial}  and plugging it back into Equation \eqref{eq:normalized-bound} completes the proof. 


\end{proof}

Next, we show that sampling a subset of size $t$ indicated by thee $w$'s proportional to the marginal pseudo-distribution on this set results in an empirical estimator that is close to $\Pi_*$ with constant probability.

\begin{lemma}
\label{lem:key_rounding_lemma}
Given $t \in \mathbb{N}$, let $\tmu$ be a pseudo-distribution of degree at least $t +2k$ satisfying $\cA_{w,\Pi}$ and minimizing $\norm{\pE_{\tmu}{[w]}}_2$. Let $S \subseteq \cI$, $|S| \leq t$ be chosen randomly with probability proportional to $\binom{\cI}{S} \pE_{\tmu}[w_S]$. Let $\tmu_S$ be the pseudo-distribution obtained by re-weighting $\tmu$ by the SoS polynomial $w_S^2$. Then, with probability at least $9/10$ over the draw of $S$, $\Norm{\pE_{\tmu_S}\left[ \Pi\right]-\Pi_*}_{F}^2 \leq \eta/10$.
\end{lemma}

\begin{proof}
Rewriting the conclusion of Lemma~\ref{lem:step-1-large-rounding}, we have:
\begin{equation}\label{eq:expectation-disguised}
\frac{1}{\pE_{\tilde{\mu}}\left[\Paren{\sum_{i \in \cI} w_i}^{t}\right]} \sum_{S \subseteq \cI, |S| \leq t} {\cI \choose S} \pE_{\tilde{\mu}}[w_S] \frac{\pE_{\tilde{\mu}}\left[w_S \Norm{\Pi-\Pi_*}_F^2\right]}{\pE_{\tilde{\mu}}[w_S]} \leq \eta/100\mper
\end{equation}
Further, $\sum_{S \subseteq \cI, |S|\leq t} {\cI \choose  S} \pE_{\tilde{\mu}}[w_S] = \pE_{\tilde{\mu}}\Paren{\sum_{i \in \cI} w_i}^t$. Thus, $\frac{{\cI \choose  S} \pE_{\tilde{\mu}}[w_S] }{\pE_{\tilde{\mu}}\Paren{\sum_{i \in \cI} w_i}^t}$ is a probability distribution, $\zeta$, over $S \subseteq \cI, |S|\leq t$. Thus, we can rewrite \eqref{eq:expectation-disguised} as simply:
\[
\E_{S \sim \zeta} \left[\frac{\pE_{\tilde{\mu}}[w_S \Norm{\Pi-\Pi_*}_F^2]}{\pE_{\tilde{\mu}}[w_S]} \right]\leq \eta/100\mper
\]
By Markov's inequality, a $S \sim \zeta$ satisfies $\frac{\pE_{\tilde{\mu}}[w_S \Norm{\Pi-\Pi_*}_F^k]}{\pE[w_S]} \leq \eta/10$ with probability at least $9/10$.
Finally, observe that by Fact~\ref{fact:reweightings}, $\pE_{\tmu_S} \Norm{\Pi-\Pi_*}_F^2 = \frac{\pE_{\tilde{\mu}}[w_S \Norm{\Pi-\Pi_*}_F^2]}{\pE_{\tilde{\mu}}[w_S]}$. Thus, with probability at least $9/10$ over the choice of $S \sim \zeta$, $\pE_{\tmu_S} \left[\Norm{\Pi-\Pi_*}_F^2\right] \leq  \eta/10$.
By Cauchy-Schwarz inequality applied with $f=1$ and $g = \|(\Pi-\Pi_*)\|^2_F$, we have: $\Norm{\pE_{\tilde{\mu}}[(\Pi-\Pi_*)]}_F^2 \leq \pE_{\tilde{\mu}} \left[\Norm{\Pi-\Pi_*}_F^2\right]$. Thus, $\Norm{\pE_{\tmu_S}\left[\Pi\right]-\Pi_*}_F^2 \leq \eta/10$. This completes the proof.

\end{proof}


\begin{proof}[Proof of Theorem ~\ref{thm:large-list-rounding}]
We begin by observing the constraint system is feasible. Consider the assignment to $w$ that indicates the true inliers, i.e. $w_i = 1$ iff $i \in \cI$ and $\Pi= \Pi_*$. It is easy to check that $\sum_{i\in[n]}w_i = \alpha n$, $w_i \Pi x_i = w_i x_i$ and $\Pi$ is a rank-$k$ projection matrix. Given $\Paren{kd\log(d)/\alpha}^{\mathcal{O}(k)}$ samples, it follows from Lemma \ref{lem:cert_hyper_sampling} that the uniform distribution over the inliers is $k$-certifiably $c$-hypercontractive with probability $1-1/\poly(d)$ and thus $\cA_{w,\Pi}$ is feasible.

We note that since $\cD$ is $k$-certifiably $(c,\delta)$-anti-concentrated, sampling $n_0 = (kd\log(d)/\alpha)^{\mathcal{O}(k)}$ suffices for the uniform distribution over $\cI$ to be $k$-certifiably $(c,2\delta)$-anti-concentrated (this follows from Lemma \ref{lem:certifiable_anti_conc_sampling}). 
We then observe that by Fact \ref{fact:high-entropy-pseudo-distributions-2}, $\pE_{\tmu}\left[w(\cI)^t\right]  = \frac{\pE_{\tmu}\left[(\sum_{i\in \cI} w_i)^t\right]}{\pE_{\tmu}\left[(\sum_{i\in [n]} w_i)^t\right] } \geq \alpha^{t}$. Therefore, with probability at least $9\alpha^t/10$, $w_S \subset \cI$ and the conclusion of Lemma \ref{lem:key_rounding_lemma} holds. From Lemma \ref{lem:key_rounding_lemma}, we can now conclude that with probability at least $9/10$, $\norm{\tilde{\Pi} - \Pi_*}^2_F \leq \eta/10$.

Finally, we analyze the sample complexity and running time. The sample complexity is dominated by the uniform distribution over the inliers being certifiable anti-concentrated, which requires $(kd\log(d)/\alpha)^{\mathcal{O}(k)}$ samples. The running time is dominated by computing a degree $O(2k + t)$ pseudo-distribution over $\cA_{w,\Pi}$ that requires $n^{O(2k + t)} = n^{O(k + \log(1/\eta))}$ time. 

\end{proof}

\subsection{Analyzing $\cA_{w,\Pi}$: Proof of Lemma ~\ref{lem:main-sosized-fact}}

We now provide a sum-of-squares proof of the polynomial inequality stated in  Lemma ~\ref{lem:main-sosized-fact}. 

\begin{lemma}[Covariance of Subsets of Certifiably Anti-Concentrated Distributions]
\label{lem:subset_cov_anti_conc}

If the uniform distribution over the inliers is $k$-certifiably $(C,\delta)$ anti-concentrated, we have
\begin{equation}
\Set{w_i^2 = w_i \mid \forall i} \sststile{2k}{w,v} \Set{\frac{1}{n} \sum_{i = 1}^n \norm{\Pi_* v}_2^{k-2} w_i\iprod{x_i,v}^2 \geq  \delta^2 \Paren{\frac{1}{n} \sum_{i = 1}^n w_i-C\delta} \norm{\Pi_* v}_2^{k}}\mcom
\end{equation}

\end{lemma}

\begin{proof}
Let $p$ be the degree $k$ polynomial provided by Definition~\ref{def:certifiable-anti-concentration-homogenous} (certifiable anti-concentration) applied to $\cI$.
Thus, for each $1 \leq i \leq n$, we must have:
\[
\sststile{2k}{v} \Set{\norm{\Pi_* v}_2^{k-2} \iprod{x_i,v}^2 + \delta^2 p^2\paren{\iprod{x_i,v}} \geq \delta^2 \norm{v}_2^{k}}\mper
\]
Observe that
\[
\Set{w_i^2 = w_i} \sststile{2}{w_i} \Set{w_i \geq 0}\mper
\]


Using the above along with \eqref{eq:sos-addition-multiplication-rule} for manipulating SoS proofs, we must have:
\[
\Set{w_i^2 = w_i \mid \forall i} \sststile{2k}{w,v} \Set{\frac{1}{n} \sum_{i = 1}^n \norm{\Pi_* v}_2^{k-2} w_i\iprod{ x_i,v}^2 + \delta^2 \frac{1}{n} \sum_{i = 1}^n w_i p^2\paren{\iprod{x_i,v}} \geq \delta^2 \frac{1}{n} \sum_{i = 1}^n w_i \norm{\Pi_* v}_2^{k}}\mper
\]
Rearranging yields:
\begin{equation}
\Set{w_i^2 = w_i \mid \forall i} \sststile{2k}{w,v} \Set{\frac{1}{n} \sum_{i = 1}^n \norm{\Pi_* v}_2^{k-2} w_i\iprod{x_i,v}^2 \geq  \delta^2 \frac{1}{n} \sum_{i = 1}^n w_i \norm{\Pi_* v}_2^{k}-\delta^2 \frac{1}{n} \sum_{i = 1}^n w_i p^2\paren{\iprod{x_i,v}} }\mper \label{eq:eq-poly-1}
\end{equation}

Next, observe that $\Set{w_i^2 = w_i} \sststile{2}{w_i} \Set{(1-w_i)= (1-w_i)^2 \geq 0}$. Thus, $\Set{w_i^2 = w_i} \sststile{2}{w_i} \Set{w_i \leq 1}$. As a consequence, $\Set{w_i^2 = w_i} \sststile{k+2}{w_i,v} \Set{w_i p^2(\iprod{x_i , v })\leq p^2(\iprod{x_i , v })}$. Summing up over $1 \leq i \leq n$ yields:

\[
\Set{w_i^2 = w_i \mid \forall i} \sststile{2k}{w,v} \Set{\frac{1}{n} \sum_{i = 1}^n w_i p^2\paren{\iprod{x_i,v}}\leq \frac{1}{n}\sum^n_{ i =1}p^2\paren{\iprod{x_i,v}} \leq C\delta \norm{\Pi_* v}_2^{k}}\mcom
\]
where in the final inequality on the RHS above, we used the second condition from Definition~\ref{def:certifiable-anti-concentration-homogenous} satisfied by $\cS$.
Plugging this back in \eqref{eq:eq-poly-1}, we thus have:

\begin{equation}
\Set{w_i^2 = w_i \mid \forall i} \sststile{2k}{w,v} \Set{\frac{1}{n} \sum_{i = 1}^n \norm{\Pi_* v}_2^{k-2} w_i\iprod{ x_i,v}^2 \geq  \delta^2 \Paren{\frac{1}{n} \sum_{i = 1}^n w_i-C\delta} \norm{\Pi_* v}_2^{k}}\mcom
\end{equation}
as desired.

\end{proof}

\begin{lemma}[Technical SoS fact about Powering]
\label{lem:powering}
For indeterminates $a,b,Z$ and any $t \in \N$,
\begin{equation}
 \{a \geq 0, b \geq 0, (a-b)Z  \leq 0 \} \sststile{t}{a,b} \left\{ (a^{t} - b^{t})Z \leq 0 \right\}
\end{equation}

\end{lemma}

\begin{proof}
We have:
\[
\Set{a \geq 0, b \geq 0} \sststile{t}{a} \Set{\sum_{i = 0}^{t-1} a^{t-1-i}b^{i} \geq 0}\mper
\]
Using the  above identity with \eqref{eq:sos-addition-multiplication-rule} yields:
\[
\Set{a \geq 0, b \geq 0, (a-b)Z \leq 0 } \sststile{t}{a,b} \left\{ \Paren{a-b}\Paren{\sum_{i = 0}^{t-1} a^{t-1-i}b^{i}} Z \leq 0 \right\}\mper
\]
Using the identity: $\Paren{a-b}\Paren{\sum_{i = 0}^{t-1} a^{t-1-i}b^{i}} = a^{t}-b^t$, we finally obtain:
\[
\Set{a \geq 0, b \geq 0, (a-b)Z \leq 0 } \sststile{t}{a,b} \left\{ \Paren{a^{t}-b^t} Z\leq 0\right\}\mper
\]

\end{proof}

\begin{proof}[Proof of Lemma \ref{lem:main-sosized-fact}]
From Lemma \ref{lem:subset_cov_anti_conc}, we have: 

\begin{equation}
    \Set{w_i^2 = w_i \text{ }\mid \forall i}  \sststile{2k}{w,v} \Set{\frac{1}{|\cI|} \sum_{i\in \cI} w_i \iprod{x_i,v}^2 \Norm{\Pi_* v}_2^{k-2} \geq \delta^2\left( w(\cI) - C\delta \right) \Norm{\Pi_* v}_2^{k}}
\end{equation}

Let $M = \II-\Pi$. Since $x_i = \Pi_* x_i$, we have the following polynomial identity  (in indeterminates $\Pi,v$) for any $i$:

\[
\iprod{ \Pi_* x_i, \Pi_*M  v} = \iprod{M x_i, v}\mper
\]

By using the \eqref{eq:sos-substitution} for manipulating SoS proofs, substituting $v$ with the polynomial $\Pi_* M  v$, and recalling $\Pi_*^2 = \Pi_*$, we thus obtain:

\begin{equation}
    \Set{\forall i \in [n] \textrm{  } w_i^2 = w_i }  \sststile{2k}{w,v} \Set{\frac{1}{|\cI|} \sum_{i\in \cI} w_i \iprod{M x_i,v}^2 \Norm{\Pi_* M  v}_2^{k-2} \geq \delta^2\left( w(\cI) - C\delta \right) \Norm{\Pi_* M  v}_2^{k}} \label{eq:pre-0-expression}
\end{equation}
Next, observe that $\cA_{w,\Pi} \sststile{2}{w,\Pi} \Set{ w_i M x_i = 0 \text{ } \forall i}$ and thus,
\[
\cA_{w,\Pi} \sststile{4}{w,v,\Pi} \Set{ \iprod{w_i M x_i,v}^2 = w_i\iprod{M x_i,v}^2 =  0 \text{ } \forall i}\mper
\]
Combining this with \eqref{eq:pre-0-expression}, we thus have:
\begin{equation}
    \cA_{w,\Pi}  \sststile{2k}{w,v} \Set{0 \geq \delta^2\left(w(\cI)  - C\delta \right) \Norm{\Pi_* M v}_2^{k}} \label{eq:post-0-expression}
\end{equation}
Using \eqref{eq:sos-addition-multiplication-rule} to multiply throughout by the constant $1/\delta^2$ yields:

\begin{equation}
     \cA_{w,\Pi} \sststile{2k}{w,v} \Set{0 \geq \left(w(\cI)  - C\delta \right) \Norm{\Pi_* M  v}_2^{k}} \label{eq:post-0-expression_2}
\end{equation}
Applying Lemma~\ref{lem:powering} with $a = \frac{1}{|\cI|} \sum_{i \in\cI} w_i$, $b = C\delta$ and $Z = \Norm{\Pi_* M  v}_2^{k}$, we obtain:
\begin{equation}
     \cA_{w,\Pi} \sststile{2k + t}{w,v} \Set{0 \geq  \Paren{w(\cI)^t - \paren{C\delta}^t} \Norm{\Pi_* M  v}_2^{k}} \label{eq:post-0-expression_3}
\end{equation}
Let $g \sim \cN(0,I)$. Then, using the above with the substitution $v = g$, we have:
\begin{multline}
\label{eq:bound_useful}
\cA_{w,\Pi} \sststile{2k}{v,w} \Biggl\{ w(\cI)^t \Paren{\E \left[g^{\top} M\Pi_*M g\right]}^{k/2} =  w(\cI)^t \Tr(M\Pi_*M)^{k/2} = w(\cI)^t (\E \left[g^{\top} M\Pi_*M g\right]^{k/2}\\
 \leq w(\cI)^t\E \left[g^{\top} M\Pi_*M g\right]^{k/2} \leq w(\cI)^t (Ck)^{k} \Paren{\E\left[g^{\top} M\Pi_*M g\right]}^{k/2} \Biggr\}\mcom
\end{multline}
where the inequality follows from the SoS Hölder's inequality.

Next,
\[
\Set{\Pi^2 = \Pi} \sststile{2}{\Pi} \Set{\Norm{\Pi}_F^2 = \Tr(\Pi^2) = \Tr(\Pi)=r}\mper
\]
And also,
\[
\Set{\Pi^2 = \Pi} \sststile{2}{\Pi} \Set{M^2 = (I-\Pi)^2 = I -2\Pi + \Pi^2 = I-\Pi = M}\mper
\]
Thus,
\begin{equation*}
\begin{split}
\cA_{w,\Pi} \sststile{2}{\Pi} \Biggl\{ \Norm{\Pi-\Pi_*}_F^2 &= \Norm{\Pi}_F^2 + \Norm{\Pi_*}_F^2 - 2 \Tr(\Pi \Pi_*) = 2r - 2 \Tr(\Pi \Pi_*) \\
&= 2 \Tr((I-\Pi)\Pi_*) = 2\Tr(M\Pi_*) = 2 \Tr(M^2\Pi_*) = 2\Tr(M\Pi_*M)\Biggr\}\mper
\end{split}
\end{equation*}
where the last equality follows from the cyclic property of the trace.

Using \eqref{eq:sos-addition-multiplication-rule} and combining with \eqref{eq:bound_useful}, we thus obtain:
\[
\cA_{w,\Pi} \sststile{2k + t}{\Pi,w} \Set{ w(\cI)^t \Norm{\Pi-\Pi_*}_F^k = w(\cI)^t 2^{k/2} \Tr(M\Pi_*M)^{k/2} \leq (C\delta)^t (Ck)^k \Norm{\Pi-\Pi_*}_F^k} \mper
\]
Multiplying both sides by $w(\cI)$ and using Lemma~\ref{lem:close-projectors} yields

\[
\cA_{w,\Pi} \sststile{2k^2 + t}{\Pi,w} \Set{ w(\cI)^{t+1} \Norm{\Pi-\Pi_*}_F^k  \leq (C\delta)^t (Ck)^k (4C^2k)^k} \mper
\]
which concludes the proof. 

\end{proof}





\section{Relevant Analytic Properties of Distributions}
\label{sec:certifiable_anti_conc}
In this section, we prove basic facts about the three relevant properties of probability distributions relevant to this work.

\subsection{Certifiable Anti-Concentration}
We start by recalling the definition again. 

\begin{definition}[Certifiable Anti-Concentration]
\label{def:certifiable-anti-concentration-homogenous-restated}
A zero-mean distribution $D$ with covariance $\Sigma$ is $2k$-certifiably $(\delta,C\delta)$-anti-concentrated if there exists a univariate polynomial $p$ of degree $d$ such that:
\begin{enumerate}
\item $\sststile{2k}{v} \Set{ \norm{v}^{2k-2}_2\iprod{x,v}^2 + \delta^2 p^2\left(\iprod{x,v}\right) \geq \frac{\delta^2 \norm{\Sigma^{1/2}v}^{2k}_2}{4}}$. 
\item $\sststile{2k}{v} \Set{\E_{ x \sim D}\left[ p^2\paren{\iprod{x,v}} \right] \leq C\delta \norm{\Sigma^{1/2} v}_2^{2k}}$.
\end{enumerate}
A set $\cS$ is $2k$-certifiably $(C,\delta)$-anti-concentrated if the uniform distribution on $\cS$ is $2k$-certifiably $(C,\delta)$-anti-concentrated.

\end{definition}

As discussed earlier, this definition is obtained by a important but technical modification of the definition used in ~\cite{DBLP:journals/corr/abs-1905-05679,raghavendra2020list}. We verify basic properties of this notion here and establish that natural distributions such as Gaussians do satisfy it.  
We first prove that natural distributions like the Gaussians and uniform distribution on the unit sphere are certifiably anti-concentrated. 

\begin{theorem}(Certifiable Anti-Concentration of Gaussians.)
\label{thm:cert_anti_conc_gaussians}
Given $0<\delta \leq 1/2$, there exists $k = O\left(\frac{\log^{5}(1/\delta)}{\delta^{2}} \right)$ such that $\cN(0, \Sigma)$ is $k$-certifiably $(C, \delta)$-anti-concentrated. 
\end{theorem}

Our proof of Theorem~\ref{thm:cert_anti_conc_gaussians} will rely on the following construction of a low-degree polynomial with certain important properties:


\begin{fact}[Core Indicator for Strictly Sub-Exponential Tails, Lemma A.1 in ~\cite{DBLP:journals/corr/abs-1905-05679}, ]
\label{lem:core_indicator} Given a distribution $\cD$ on $\R^d$ with mean $0$ and variance $\sigma\leq 1$ satisfying:
\begin{enumerate}
  \item \textbf{Anti-Concentration:} for all $\eta >0$, $\Pr_{x \sim \cD}[|x| \leq \eta \sigma] \leq c_1 \eta$,
  \item \textbf{Strictly Sub-Exponential Tail:} for some $k<2$, $\Pr_{x\sim\cD}[|x| \geq t\sigma] \leq \exp(-t^{2/k}/c_2)$,
\end{enumerate}
for some fixed $c_1, c_2 > 1$. Then, for any $\delta > 0$, there exists a degree $O\left(\frac{\log^{(4+k)/(2-k)}(1/\delta)}{\delta^{2/(2-k)}} \right)$ \textit{even} polynomial $q$ with all coefficients upper-bounded by $t^{O(k)}$ satisfying:
\begin{enumerate}
\item $|x| \leq \delta$, $q(x)= 1\pm \delta$, and,
\item $\sigma^2 \expecf{x\sim \cD}{q^2(x)}\leq 10 c_1 c_2 \delta$.
\end{enumerate}
\end{fact}

We will also use the following basic fact about even polynomials.

\begin{lemma}[Structure of Even Polynomials]\label{fact:struct_of_even_poly} For any even univariate polynomial $q()$ of degree $t$, $\norm{v}^{2t}_2q^2(\iprod{x,v}/\norm{v}_2)$ is a polynomial in vector-valued indeterminate $v$ and further, \[\sststile{2t}{v} \Set{ \norm{v}^{2t}_2q^2(\iprod{x,v}/\norm{v}_2) \geq 0}. \]
\end{lemma}

\begin{proof}
The conclusion requires us to prove that $\norm{v}^{2t}_2 q^2(\iprod{x,v}/\norm{v}_2)$ is a sum-of-squares polynomial in vector-valued variable $v$.   
Let $q(z) = \sum_{i\in d} c_i z^{i}$. Since $q(z)$ is even,  
\[
q(z) = \frac{1}{2}(q(z)+q(-z)) = \frac{1}{2} \left( \sum_{i \in [d]} c_i z^{i} + c_i (-z)^{i} \right) =  \sum_{1\leq i \leq t/2} c_{2i} z^{2i} \mper
\]
Thus, in particular, $t$ is even and $q() =r(z^2)$ for some polynomial $r$ of degree $t/2$. Substituting $z = \iprod{x,v}/\norm{v}_2$, we have; $\norm{v}^{2t}_2q^2(\iprod{x,v}/\norm{v}_2) = \norm{v}^{2t}_2 \left(\sum_{i\leq d/2} c_{2i} \frac{\iprod{x,v}^{2i}}{\norm{v}^{2i}_2}\right)^2 = \left(\sum_{i\leq d/2} c_{2i} \norm{v}^{t-2i}_2  \iprod{x,v}^{2i}\right)^2$ which is a sum-of-squares polynomial in $v$. 
\end{proof}

We can now prove Theorem \ref{thm:cert_anti_conc_gaussians}.

\begin{proof}[Proof of Theorem \ref{thm:cert_anti_conc_gaussians}] 
Let $x \sim \cN(0, \Sigma)$. 
We begin with the following polynomial :
\[
p(v) = \norm{v}^{2t}_2 q(\iprod{\Sigma^{\dagger/2} x,v}/\norm{v}_2)
\]
where $q$ is the degree $2t = \Theta\left(\frac{\log^{5}(1/\delta)}{\delta^{2}} \right)$ polynomial from Lemma \ref{lem:core_indicator}. By Fact \ref{fact:struct_of_even_poly}, $p$ is indeed a univariate polynomial in $v$. We will prove that $\cN(0,\Sigma)$ is $2t$-certifiably $(C,\delta)$-subgaussian for some some absolute constant $C > 0$ using the polynomial $p$.

Consider the polynomial $g(x) = x^2 + \delta^2 q^2(x) - \delta^2/4$. If $|x| > \delta$ then, $g(x) \geq 3\delta^2/4 \geq 0$. On the other hand, if $|x| \leq \delta$, using that $q^2(x) = (1\pm\delta)^2\geq \frac{1}{4}$ for every $\delta \leq 1/2$, $g(x) \geq 0$. Thus, $g$ is a univariate, non-negative polynomial. Using Fact \ref{fact:univariate} we thus obtain:

\begin{equation*}
\sststile{2t}{x} \Set{ x^2 + \delta^2 q^2(x) \geq \delta^2/4}\mcom
\end{equation*}
or, equivalently, $x^2 + \delta^2 q^2(x) - \delta^2/4 = s(x)$ for a SoS polynomial $s$ of degree at most $2t$. Since $q$ is even, the LHS is invariant under the transformation $x \rightarrow -x$. Thus, $s$ is an even polynomial.

Substituting $x = \frac{\iprod{\Sigma^{\dagger/2}x,v}}{\norm{v}_2}$, we thus have:  
\begin{equation*}
\frac{\iprod{\Sigma^{\dagger/2}x,v}^2}{\norm{v}^2_2} + \delta^2 q^2\left(\frac{\iprod{\Sigma^{\dagger/2}x,v}}{\norm{v}_2}\right) - \frac{\delta^2}{2}  = s\left(\frac{\iprod{\Sigma^{\dagger/2}x,v}}{\norm{v}_2}\right)
\end{equation*}
Multiplying out by $\norm{v}^{2t}_2$ and using the definition of $p$ gives us the polynomial (in $v$) identity:
\begin{equation*}
\norm{v}^{2t-2}_2 \iprod{\Sigma^{\dagger/2}x,v}^2 + \delta^2 p^2\left(\iprod{\Sigma^{\dagger/2}x,v}\right) - \frac{\delta^2 \norm{v}^{2t}_2}{4}  = \norm{v}^{2t}_2 s\left(\frac{\iprod{\Sigma^{\dagger/2}x,v}}{\norm{v}_2}\right)
\end{equation*}
Since $s$ is an even polynomial, it follows from Fact \ref{fact:struct_of_even_poly}, $\norm{v}^{2t}_2 s\left(\frac{\iprod{\Sigma^{\dagger/2}x,v}}{\norm{v}_2}\right)$ is a sum-of-squares in $v$. Thus, we can conclude:
\begin{equation*}
\sststile{2t}{v} \Set{ \norm{v}^{2t-2}_2\iprod{\Sigma^{\dagger/2}x,v}^2 + \delta^2 p^2\left(\iprod{\Sigma^{\dagger/2}x,v}\right) \geq \frac{\delta^2 \norm{v}^{2t}_2}{4}}
\end{equation*}
which completes the proof of the first inequality in Definition \ref{def:certifiable-anti-concentration-homogenous}. By rotational invariance of Gaussians, $\expecf{x\sim \cN(0,1)}{\iprod{x,v}^{\ell}}$ is just a function of $\norm{v}^{2\ell}_2$. Thus $\norm{v}^{2t}_2\expecf{x\sim \cN(0, \Sigma)}{ q^2\left(\frac{\iprod{\Sigma^{\dagger/2}x,v}}{\norm{v}}\right)}$ is a polynomial in $\norm{v}^2_2$. Since $\Sigma^{\dagger/2} x$ has variance $1$, it follows from the definition of $p$ and $q$ that $\E_{ x \sim D} p^2\paren{\iprod{\Sigma^{\dagger/2} x,v}} \leq C\delta \norm{v}_2^{2t}$, for $C= 10 c_1 c_2$. Therefore, applying Fact \ref{fact:univariate}
\begin{equation*}
\sststile{2t}{\norm{v}^2_2} \Set{\E_{ x \sim D} p^2\paren{\iprod{\Sigma^{\dagger/2} x,v}} \leq C\delta \norm{v}_2^{2t}}
\end{equation*}
\end{proof}

The proof above naturally extends to the uniform distribution on the unit sphere.

\begin{theorem}(Certifiable Anti-Concentration of Gaussians.)
\label{thm:cert_anti_conc_uniform}
Given $0<\delta \leq 1/2$, there exists $k = O\left(\frac{\log^{5}(1/\delta)}{\delta^{2}} \right)$ such that the uniform distribution on the unit sphere is $k$-certifiably $(C, \delta)$-anti-concentrated. 
\end{theorem}

Next, we observe that our definition of certifiable anti-concentration is invariant under linear transformations: 

\begin{lemma}(Linear Invariance.) Let $x \sim \cD$ such that $\cD$ is $k$-certifiably $(C, \delta)$-anti-concentrated distribution. Then, for any $A \in \R^{m \times d}$, the random variable $Ax$ has a $k$-certifiably $(C, \delta)$-anti-concentrated distribution.
\end{lemma}

As a simple corollary, we observe that certifiable-anti-concentration is preserved under taking linear projections of a distribution.

\begin{corollary}(Closure under taking projections) 
Let $x \sim \cD$ such that $\cD$ is $k$-certifiably $(C, \delta)$-anti-concentrated distribution on $\R^d$. Let $V$ be any subspace of $\R^d$ and let $\Pi_V$ be the associated projection matrix. Then, the random variable $\Pi_V x$ has a $k$-certifiably $(C, \delta)$-anti-concentrated distribution.
\end{corollary}

Next, we show that anti-concentration is preserved under sampling, i.e. if $\cD$ is anti-concentrated, then the uniform distribution over $n$ samples from $\cD$ is also anti-concentrated. 

\begin{lemma}(Certifiable Anti-Concentration under Sampling.)
\label{lem:certifiable_anti_conc_sampling}
Let $\cD$ be $k$-certifiably $(c, \delta)$-anti-concentrated sub-exponential distribution such that the certifying polynomial $p$ has coefficients bounded by $d^{O(k)}$. Let $\cS$ be a set of $n = \Omega( (kd\log(d))^{O(k)}/C\delta)$ i.i.d. samples from $\cD$. Then, with probability at least $1-1/d$, the uniform distribution on $\cS$ is $k$-certifiably $(2c,\delta)$-anti-concentrated.
\end{lemma}
\begin{proof}
Let $p$ be a degree-$k$ that witnesses anti-concentration of $\cD$. We show that $p$ also witnesses anti-concentration of the uniform distribution on $\cS$, denoted by $\cD'$. First, we observe that property 1 in definition \ref{def:certifiable-anti-concentration-homogenous} is point-wise and continues to hold for $x$ sampled from $\cD'$. 

Further, we know that 
\begin{equation}
\label{eqn:cert_anti_conc_statement}
\sststile{2k}{v} \Set{\E_{ x \sim \cD}\left[ p^2\paren{\iprod{\Sigma^{\dagger/2} x,v}} \right] \leq C\delta \norm{v}_2^{2k}}
\end{equation}
Since $p^2$ is a square polynomial, we can represent it in the monomial basis as $p^2\paren{\iprod{\Sigma^{\dagger/2} x,v}}= \iprod{c(\Sigma^{\dagger/2}x)c(\Sigma^{\dagger/2}x)^{\top}, (1,v)_{\otimes k}  (1,v)^{\top}_{\otimes k}}$, where $c(\Sigma^{\dagger/2}x)$ are the coefficients of $\E_{x\sim \cD} p\paren{\iprod{\Sigma^{\dagger/2} x,v}}$ and $(1,v)_{\otimes k}$ are all monomials of degree at most $d$. For notational convenience let $c_x = c(\Sigma^{\dagger/2}x)$.

Since $\E_{x\sim \cD} p^2\paren{\iprod{\Sigma^{\dagger/2} x,v}} =  \iprod{\E_{x\sim \cD} c_x c_x^{\top}, (1,v)_{\otimes k}  (1,v)^{\top}_{\otimes k}}$, and $\E_{x\sim \cD} c_x c_x^{\top} = \E_{x\sim \cD'} c_x c_x^{\top} + \left(\E_{x\sim \cD} c_x c_x - \E_{x\sim \cD'} c_x c_x^{\top}\right)$, using linearity of expectation and the \eqref{eq:sos-substitution} in \eqref{eqn:cert_anti_conc_statement},

\begin{equation}
\label{eqn:cert_anti_conc_under_sampling}
\begin{split}
& \sststile{2k}{v} \Set{ \iprod{\E_{x\sim \cD'} c_x c_x^{\top}, (1,v)_{\otimes k}  (1,v)^{\top}_{\otimes k}} + \iprod{|\E_{x\sim \cD} c_x c_x^{\top} - \E_{x\sim \cD'} c_x c_x^{\top}|, (1,v)_{\otimes k}  (1,v)^{\top}_{\otimes k}}  \leq C\delta \norm{v}_2^{2k}}\\
& \sststile{2k}{v} \Set{ \iprod{\E_{x\sim \cD'} c_x c_x^{\top}, (1,v)_{\otimes k}  (1,v)^{\top}_{\otimes k}}  \leq \Norm{\E_{x\sim \cD} c_x c_x^{\top} - \E_{x\sim \cD'} c_x c_x^{\top}}_2 \norm{v}^{2k}_2 + C\delta \norm{v}_2^{2k}}
\end{split}
\end{equation}
where the second equation follows from Fact \ref{fact:operator_norm}.
Observe, it suffices to bound each entry, i.e. for all $i, j \in [d^{2k}]$, $\left(\expecf{x\sim\cD'}{(c_x)_i (c_x)_j} - \expecf{x\sim\cD}{(c_x)_i (c_x)_j}\right)^2 \leq (C^2\delta^2/d^{2k})$, with probability at least $1-1/d$. 
Then, using concentration of polynomials of Sub-exponential random variables, for all $i,j \in [d^k]$, 
\begin{equation*}
\begin{split}
 \Pr_{x\sim \cD}\Bigg[ \Big(\expecf{x\sim\cD'}{(c_x)_i (c_x)_j} & -  \expecf{x\sim\cD}{(c_x)_i (c_x)_j}\Big)^2  > \epsilon^2 \Bigg] \\
& \leq \exp\left(-\left(\frac{ \epsilon n }{\expec{x\sim\cD}{(c_x)_i (c_x)_j}^2}\right)^{\frac{1}{2k}} \right)
\end{split}
\end{equation*}
Setting $\epsilon = C \delta/d^{2k}$ and union bounding over all $i$ and $j$, 
\begin{equation*}
\begin{split}
 \Pr\Bigg[ \sum_{i,j\in[d^k]} \Big(\expec{\cD'}{(c_x)_i (c_x)_j} &- \expec{\cD}{(c_x)_i (c_x)_j}\big)^2 > C^2 \delta^2 \Bigg]\\
& \leq d^{2k} \exp\left(-\left(\frac{ n }{d^{O(k)}}\right)^{\frac{1}{2k}} \right)
\end{split}
\end{equation*}
where the bound on $\expec{x\sim\cD}{(c_x)_i (c_x)_j}^2$ follows from our assumption on the coefficients of $p$. 
Setting $n =\Omega( (kd\log(d))^{O(k)}/C\delta )$ suffices to bound the above probability by $1/d$. 
\end{proof}

\subsection{Certifiable Hypercontractivity of Degree $2$ Polynomials.}
We next recall basic facts about certifiable hypercontractivity of degree $2$ polynomials.
\begin{definition}[Certifiable Hypercontractivity]
Let $\cD$ be a distribution on $\R^d$. For an even integer $h$, $\cD$ is said to have $h$-certifiably $C$-hypercontractive degree $2$ polynomials if for the $d \times d$ matrix-valued indeterminate $P$, 
\[
\sststile{h}{P} \Set{
\E_{x\sim \cD} \Paren{x^{\top} P x -\E_{x\sim \cD}x^{\top} P x}^{h} \leq  (Ch)^{h} \Paren{\E_{x\sim \cD} \Paren{x^{\top} P x -\E_{x\sim \cD}x^{\top} P x}^{2}}^{h/2}}\mper\]
\end{definition}

Observe that certifiable hypercontractivity is invariant under linear transformations. 

\begin{fact}[Certifiable Hypercontractivity of Product Distributions, Theorem 1.4~\cite{MR3376479-Kauers14}]
\label{fact:certifiable_hypercontractivity_gaussians}
Let $\cD$ be a $0$ mean, $2h$-wise independent distribution with $C$-subgaussian coordinate marginals. Then, each $\cD$ has $h$-certifiably $O(C)$-hypercontractive degree $2$ polynomials. 
\end{fact}

As an immediate consequence, Gaussian distribution and linear transformations of uniform distributions on $\{-1,1\}^d$ and $\{-q,-q+1,\ldots,0,\ldots,q\}^{d}$ are $O(1)$-certifiably hypercontractive. 

Next, we recall that uniform distribution on a large enough sample from a certifiably hypercontractive distribution is certifiably hypercontractive. 

\begin{lemma}[Certifiable Hypercontractivity Under Sampling, Lemma 8.3 \cite{bakshi2020outlierrobust}]
\label{lem:cert_hyper_sampling}
Let $\cD$ be a $1$-subgaussian, $2h$-certifiably $c$-hypercontractive distribution over $\mathbb{R}^d$. Let $\cS$ be a set of $n = \Omega( (hd)^{8h})$ i.i.d. samples from $\cD$. Then, with probability at least $1-1/\poly(n)$, the uniform distribution on $\cS$ is $h$-certifiably $(2c)$-hypercontractive.
\end{lemma}

\paragraph{Bounded Variance of Degree-$2$ Polynomials.}

Next, we recall some basic facts from~\cite{bakshi2020outlier} about certifiable bounded variance of degree $2$ polynomials of various distributions. 

Recall that we say that a zero mean distribution $\cD$ with covariance $\Sigma$ has certifiably $C$-bounded variance degree $2$ polynomials if $\sststile{2}{Q} \Set{\E_{x \sim \cD} (x^{\top}Qx-\E_{x \sim \cD} x^{\top}Qx)^2 \leq C \Norm{\Sigma^{1/2}Q\Sigma^{1/2}}_F^2}.$ Unlike in the case of hypercontractivity and anti-concentration, for the property of having $C$-bounded variance degree $2$ polynomials, the demand of a degree $2$ SoS proof places no additional restriction. Indeed, since the property is an unconditional quadratic inequality in the indeterminate $Q$, by Fact~\ref{fact:spectral-sos}, there's a degree $2$ SoS proof of it whenever there is \emph{any} proof of it. 

We first observe that $4$-wise independent distributions have $3$-bounded variance degree $2$ polynomials.

\begin{lemma}[Bounded Variance of Degree 2 Polynomials of 4-wise independent distributions]
Let $\cD$ be an isotropic (i.e. $0$ mean, $I$-covariance), 4-wise independent distribution on $\R^d$. Then, $\cD$ has certifiably $3$-bounded variance degree $2$ polynomials. That is, 
\[
\sststile{2}{Q} \Set{\E_{\cD} \Paren{x^{\top}Qx - \E_{\cD} x^{\top}Qx}^2 \leq 3 \Norm{Q}_F^2}\mper
\]
\end{lemma}
\begin{proof}
By viewing $xx^{\top}$ and $I \in \R^{d \times d}$ as $d^2$ dimensional vectors, and using that $\E_{y \sim \cD} (yy^{\top}-I)(yy^{\top}-I)^{\top} \preceq 3I \otimes I$ for any $4$-wise independent, isotropic distribution, we have:
\begin{multline}
\sststile{2}{Q} \Biggl\{ \E_{\cD} \Paren{x^{\top}Qx - \E_{\cD} x^{\top}Qx}^2 = \E_{\cD} \Iprod{xx^{\top}-I,Q}^2 \leq \Norm{\E_{x \sim \cD}(xx^{\top}-I)(xx^{\top}-I)^{\top}}_2 \Norm{Q}_F^2 \\\leq 3\Norm{I \otimes I}_2 \Norm{Q}_F^2 = 3 \Norm{Q}_F^2 \Biggr\}\mper
\end{multline}




\end{proof}

The uniform distribution on $\sqrt{d}$-radius sphere in $d$ dimensions is not $4$-wise independent. However, the above proof only requires that $\E (y^{\otimes 2}-I)(y^{\otimes 2}-I)^{\top} \preceq C I \otimes I$. For the uniform distribution on the sphere, notice that $i,j,k,\ell$-th entry of this matrix is non-zero iff the indices are in have two repeated indices and in that case, by negative correlation of the $x_i^2$ and $x_j^2$ on the sphere, it holds that $\E x_i^2 x_j^2 \leq 1$. Thus, $\E (y^{\otimes 2}-I)(y^{\otimes 2}-I)^{\top} \preceq 3 I \otimes I$ for $y$ uniformly distribution on the $\sqrt{d}$-radius unit sphere. The above proof thus also yields:

\begin{corollary}
Let $y$ be uniform on $\sqrt{d}$-radius sphere in $d$ dimensions. Then, $y$ has certifiably $3$-bounded variance degree $2$ polynomials. 
\end{corollary}

\begin{lemma}[Linear Invariance] \label{lem:linear-invariance-bounded-variance}
Let $x$ be a random variable with an isotropic distribution $\cD$ on $\R^d$ with certifiably C-bounded variance degree $2$ polynomials. Let $A \in \R^{d \times d}$ be an arbitrary PSD matrix. Then, the random variable $x'=Ax$ also has certifiably C-bounded variance degree $2$ polynomials. 
\end{lemma}

\begin{proof}
The proof follows by noting that ${x'}^{\top}Qx' = (Ax)^{\top}Q(Ax) = x^{\top} (AQA) x^{\top}$ and covariance of $x'$ is $A^2$.
\end{proof}

\begin{lemma}[Bounded Variance Under Sampling]
Let $\cD$ be have degree $2$ polynomials with certifiably C-bounded variance and be $8$-certifiably $C$-subgaussian. Let $X$ be an i.i.d. sample from $\cD$ of size $n \geq n_0 = $. Then, the uniform distribution on $X$ has degree $2$ polynomials with certifiable $2C$-bounded variance. 
\end{lemma}

\begin{proof}
Using Lemma~\ref{lem:linear-invariance-bounded-variance}, we can assume that $\cD$ is isotropic. Arguing as in the proof of Lemma~\ref{lem:cert_hyper_sampling}, it is enough to upper-bound the spectral norm $\Norm{\frac{1}{n} (x_i^{\otimes 2}-I)(x_i^{\otimes 2}-I)^{\top}-\E_{x \sim \cD} (x^{\otimes 2}-I)(x^{\otimes 2}-I)^{\top}}_2$ by $C$. We do this below:

By applying certifiable $C$-bounded variance property to $Q = vv^{\top}$ where $e_i$ are standard basis vectors in $\R^d$, we have that $\E (\iprod{x_i,v}^2-\E \iprod{x_i,v}^2)^2 \leq C \Norm{v}_2^4$ and thus, $\E \iprod{x_i,v}^4 \leq (1+C) \Norm{v}_2^4$. By an application of the AM-GM inequality, we know that for every $i,j,k,\ell$, $(\iprod{x,e_i}^2\iprod{x,e_j}^2 \iprod{x,e_k}^2 \iprod{x,e_\ell})^2 \leq \iprod{x,e_i}^8 + \iprod{x,e_j}^8 + \iprod{x,e_k}^{8} + \iprod{x,e_\ell}^8$. Thus, the variance of every entry of the matrix $\E x^{\otimes 4}$ is bounded above by $4(C8)^4 = O(C^4)$. Thus, by Chebyshev's inequality, any given entry of $\frac{1}{n} x_i^{\otimes 4}-\E_{x \sim \cD} x^{\otimes 4}$ is upper-bounded by $O(C^2)d^4/\sqrt{n}$ with probability at least $1-1/(100d^4)$. By a union bound, all entries of this tensor are upper-bounded by $O(C^2)d^4/\sqrt{n}$ with probability at least $0.99$. Thus, the Frobenius norm of this tensor is at most $d^8 O(C^2)/\sqrt{n}$. Since $n \geq n_0 = d^8/O(C^2)$, this bound is at most $C/2$. Thus, we obtain that with probability at least $0.99$, $\Norm{\frac{1}{n} (x_i^{\otimes 2}-I)(x_i^{\otimes 2}-I)^{\top}-\E_{x \sim \cD} (x^{\otimes 2}-I)(x^{\otimes 2}-I)^{\top}}_2 \leq 2\Norm{\frac{1}{n} x_i^{\otimes 4}-\E_{x \sim \cD} x^{\otimes 4}}_F \leq C$. 

\end{proof}

The above three lemmas immediately yield that Gaussian distributions, linear transforms of uniform distribution on unit sphere, discrete product sets such as the Boolean hypercube and any 4-wise independent zero-mean distribution has certifiably $C$-bounded variance degree $2$ polynomials.

  \phantomsection
  \addcontentsline{toc}{section}{References}
  \bibliographystyle{amsalpha}
  \bibliography{bib/scholar,bib/mathreview,bib/dblp,bib/custom}

\appendix
\clearpage

\section{Exact Recovery via Anti-Concentration}
\label{sec:exact_recovery}

\begin{proposition} \label{prop:intro-high-entropy}
Let $\mu$ be a distribution on $(w,\Pi)$ satisfying $\cA_{w,\Pi}$ such that $\Norm{\sum_{i =1}^n\E_{\mu} w_i}_2^2$ is minimized. Then, $\E_{\mu}\left[ \sum_{i \in \cI} w_i \right]\geq \alpha |\cI|$.
\end{proposition}

\begin{proposition}[High Intersection Implies Same Subspace (TV Distance to Parameter Distance)] \label{prop:high-intersection-same-subspace}
Let $\cS$ be a sample of size $n$ from $\sub_{\cD}(\alpha, \Pi_*)$ where $\Pi_*$ is a projection matrix of rank $k$ such that the inliers $\cI$ are $\alpha$-anti-concentrated. Let $T \subseteq \cS$ be a subset of size $\alpha n$ such that $\Pi x = x$ for every $x \in T$ for some projection matrix $\Pi$ of rank $r$. Suppose $|T \subseteq \cI| \geq \alpha |\cI|$. Then, $\Pi = \Pi^*$.
\end{proposition}

\begin{proof}
Let $I-\Pi = \sum_{i =1}^{d-r} v_i v_i^{\top}$ for an orthonormal set of vectors $v_i$s. Since $\Pi x = x$ for every $x \in T$, $\iprod{x,v_i} = 0$ for every $x \in T$. Thus, $\Pr_{x \sim \cI} [ \iprod{x,v_i}=0] \geq |T\cap \cI|/|\cI| \geq \alpha$. Since $\cI$ is $\alpha$-anti-concentrated, this must mean that $v_i^{\top} \Pi_* v_i = 0$.

Thus, $\sum_{i} v_i^{\top}\Pi_*v_i = \tr(\Pi_* \sum_{i =1}^{d-r} v_i v_i^{\top}) = \tr(\Pi_* (I-\Pi))= 0$. Or $\tr(\Pi_*) = \tr(\Pi \cdot \Pi_*)$. On the other hand, by Cauchy-Schwarz inequality, $\tr(\Pi \cdot \Pi_*) \leq \sqrt{\tr(\Pi^2) \tr((\Pi_*)^2)} = \tr(\Pi)$ with equality iff $\Pi = \Pi_*$. Here, we used the facts that $\Pi = \Pi^2, (\Pi_*)^2 = \Pi_*$ and that $\tr(\Pi) = \tr(\Pi_*) = r$. 
Thus, $\Pi = \Pi^*$.
\end{proof}

We can use the lemma above to give an \emph{inefficient} algorithm  for list-decodable subspace recovery.
\begin{lemma}[Identifiability for Anti-Concentrated inliers] \label{lem:identifiability-via-anti-concentration}
Let $\cS$ be a sample drawn according to $\sub_{\cD}(\alpha,\Pi_*)$ such that the inliers $\cI$ are $\delta$-anti-concentrated for $\delta < \alpha$. Then, there is an (inefficient) randomized algorithm that finds a list $L$ of projectors of rank $k$ of size $20/(\alpha-\delta)$ such that $\Pi^* \in L$ with probability at least $0.99$.
\end{lemma}

\begin{proof}
Let $\mu$ be any maximally uniform distribution over soluble subset-projection pairs $(w,\Pi)$ where $w$ indicates a set $S$ of size at least $\alpha n$. For $t = 20/(\alpha-\delta)$, let $(S_1,\Pi_1), (S_2,\Pi_2), \ldots, (S_t,\Pi_t)$ be i.i.d. samples from $\mu$. Output $\{\Pi_1, \Pi_2,\ldots, \Pi_t\}$. To finish the proof, we will show that there is an $i$ such that $|S_i \cap \cI| \geq \frac{\alpha+\delta}{2}|\cI| > \delta|\cI|$. Then, we can then apply Proposition \ref{prop:high-intersection-same-subspace} to conclude that $\Pi_i = \Pi_*$.

By Proposition~\ref{prop:intro-high-entropy}, $\E_{S\sim \mu} |S \cap \cI| \geq \alpha |\cI|$. Thus, by averaging, $\Pr_{S \sim \mu} [|S \cap \cI| \geq \frac{\alpha+\delta}{2} |\cI|] \geq \frac{\alpha-\delta}{2}|\cI|$. Thus, the probability that at least one of $S_1, S_2, \ldots S_t$ satisfy $|S_i \cap \cI| \geq \frac{\alpha+\delta}{2} |\cI|$ is at least $1- (1-\frac{\alpha-\delta}{2})^{t} \geq 0.99$.
\end{proof}


\end{document}